\documentclass[letterpaper,twocolumn,10pt]{article}
\usepackage{usenix}
\usepackage[linesnumbered,ruled,vlined]{algorithm2e}
\SetKwInput{Input}{Input}
\SetKwInput{Output}{Output}
\usepackage{amsmath,amssymb,amsfonts}
\usepackage{tikz}
\usepackage{amsmath}
\usepackage{multirow}
\usepackage{array} 
\newcolumntype{P}[1]{>{\centering\arraybackslash}p{#1}}
\usepackage[skip=10pt]{caption}
\usepackage{graphicx}
\usepackage{subcaption}
\usepackage{booktabs} 
\usepackage{hyperref}
\usepackage[capitalize]{cleveref} 
\usepackage{xspace}
\usepackage[table]{xcolor} 
\newcommand{\bluetri}[1]{\textcolor{blue}{\(\triangleright\)~#1}}
\usepackage{xurl}
\usepackage{pifont}
\usepackage{amsthm}
\newtheorem{proposition}{Proposition}
\newtheorem{assumption}{Assumption}
\newtheorem{lemma}{Lemma}
\usepackage[most]{tcolorbox}
\usepackage{xcolor}

\newtcolorbox{takeawaybox}[1]{
    colback=gray!12,
    colframe=gray!12,
    boxrule=0pt,
    borderline west={4pt}{0pt}{gray!55},  
    sharp corners,
    boxsep=3pt,                
    left=7pt,                  
    right=5pt,                 
    top=4pt,                   
    bottom=4pt,                
    enhanced,
}

\newcommand{\cmark}{\ding{51}}
\newcommand{\xmark}{\ding{55}}
\newcommand{\high}{\textbf{High}}
\newcommand{\med}{Med.}
\newcommand{\low}{Low}

\newcommand{\method}{\textsc{MarkNull}\xspace}
\newcommand{\methodA}{\textsc{MarkNull-A}\xspace}

\crefname{assumption}{assumption}{assumptions}
\Crefname{assumption}{Assumption}{Assumptions}
\usepackage{filecontents}
\usepackage{booktabs}

\usepackage[available]{usenixbadges}
  
\begin{document}

\date{}

\title{\Large \bf \method: Model-Agnostic Watermark Removal in AI-Generated Images via On-Manifold Latent Manipulation}


\author{
{\rm Jie Cao$^1$, Qi Li$^1$, Zelin Zhang$^1$, Xiaodong Wu$^1$, Lingshuang Liu$^2$, Xiangman Li$^1$, Jianbing Ni$^{1,*}$}\\
$^1$Queen's University, Canada~~~~~
$^2$University of Waterloo, Canada\\
}

\maketitle
\pagestyle{empty}

\begin{abstract}

Digital watermarking has emerged as a critical technique for provenance and copyright attribution in AI-generated imagery, yet its robustness against realistic, model-agnostic removal attacks remains poorly explored. Existing attacks either succeed only against specific generative models or achieve removal at the cost of severe visual degradation.
In this paper, we propose \method, a \textit{model-agnostic} watermark removal attack via on-manifold latent manipulation. \method is grounded in a key observation: watermarked images exhibit a strong statistical dependency between the generated latent representation and the embedded initial noise. To quantify this dependency, we introduce the Noise-Latent Alignment Score (NLAS) and formulate an optimization objective that selectively decorrelates the latent representation from the embedded watermark while preserving semantic fidelity. 
Extensive evaluations across different categories of watermarking paradigms, including post-hoc, fine-tuning-based, and initial-noise-based schemes, demonstrate that \method reduces average bit accuracy to 53.14\%, approaching random-guessing (50\%), without perceptible image degradation. To further improve scalability, we propose \methodA, an amortized, optimization-free variant that distills the attack into a single forward pass, achieving 0.50 s/image with modest computational overhead. Notably, our attacks successfully compromise Google’s SynthID-Image system while preserving high visual quality and transfer effectively to video watermarking. Finally, we present an attack detection mechanism as a defensive counterpart to \method and \methodA, highlighting the necessity of developing watermark designs resilient to model-agnostic latent-space attacks.
\end{abstract}

\section{Introduction}

\begingroup
\renewcommand\thefootnote{}
\footnotetext{* Corresponding author.}
\endgroup

Generative AI, particularly Text-to-Image (T2I) models, has fundamentally transformed content creation by enabling the synthesis of highly realistic and semantically rich images. As AI-generated content becomes increasingly indistinguishable from natural imagery, it amplifies risks such as misinformation, deepfakes, and intellectual property infringement. To mitigate the misuse of generative models and synthetic media, digital watermarking has emerged as a key mechanism for copyright verification and content attribution in AI-generated imagery~\cite{cao2025secure}. As an active defense strategy, watermarking embeds imperceptible signals into generated images to support reliable provenance tracking and the detection of unauthorized usage. A prominent real-world deployment is DeepMind’s SynthID~\cite{gowal2025synthid}, a commercial watermarking system designed to embed invisible watermarks into AI-generated images for robust copyright attribution.

\begin{figure}[t]
    \centering
    \includegraphics[width=1\linewidth]{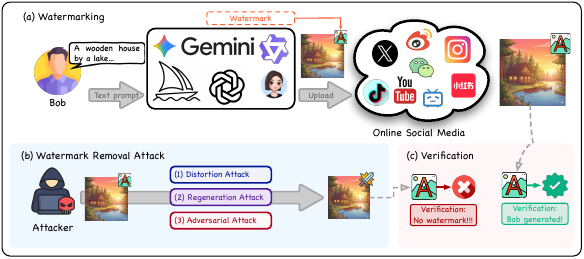}
    \caption{Overview of AI-generated image watermarking and verification, and representative removal attacks.}
    \label{fig:background}
\end{figure}

Existing watermarking techniques can be broadly classified into post-hoc and in-generation approaches. Post-hoc methods, such as DwtDctSvd~\cite{navas2008dwt}, HiDDeN~\cite{zhu2018hidden}, RivaGan~\cite{zhang2019robust}, and SynthID~\cite{gowal2025synthid}, embed watermarks directly into images after generation, including both AI-generated and real-world images~\cite{lu2024robust}. These techniques operate purely in the image domain and remain largely decorrelated from the generative model’s internal representations, limiting their robustness and reliability under realistic watermark removal attacks. To address these limitations, recent research has shifted toward in-generation watermarking, which integrates watermark embedding directly into the image generation process. By aligning watermarks with the model’s semantic and latent representations, this paradigm improves coherence and enables more reliable provenance tracking. In-generation methods can be further divided into two categories: (i) fine-tuning-based watermarking, which modifies model parameters to condition watermark presence in the generated output~\cite{fernandez2023stable,kim2024wouaf}; and (ii) initial-noise-based watermarking, which encodes watermarks into the sampling initialization of diffusion models~\cite{yang2024gaussian,wen2023tree,gunnPRC}. An overview of these paradigms is illustrated in \Cref{fig:background}.

The central challenge of image watermarking lies in robustness, i.e., the ability of watermarks to remain verifiable under distortions and malicious manipulations. While many state-of-the-art (SOTA) image watermarking methods focus on improving watermark verifiability against removal attacks, their security assessment is mainly conducted under white-box assumptions, with a particular focus on pixel-space watermarking. Therefore, critical \textbf{research gaps} remain insufficiently addressed. \ding{182} \textbf{Fidelity-removal Trade-off.} Strong watermark removal attacks can significantly reduce watermark detectability, but often introduce severe artifacts and semantic inconsistencies; in contrast, fidelity-preserving manipulation tends to be too weak to reliably evade watermark verification. For example, NFPA~\cite{qiufuture} removes watermarks via next-frame prediction, but induces object misalignment, making it difficult to preserve visual fidelity and content consistency before and after watermark removal. \ding{183} \textbf{Limited Model-agnosticism.} Existing watermark removal attacks exhibit poor generalization across watermarking schemes. Most are built upon strong assumptions about the embedding mechanism, effectively overfitting to specific watermark patterns (e.g., high-frequency artifacts in the pixel domain). These methods transfer poorly to unseen or SOTA watermarking protocols. For example, UnMarker~\cite{kassis2025unmarker} claims universality but operates exclusively in pixel space via adversarial filtering, rendering it ineffective against initial-noise-based watermarking schemes that embed watermark signals directly within the generative process. \ding{184} \textbf{Heavy Computational Overhead.} SOTA removal attacks such as Imprint~\cite{muller2025black} and CtrlRegen$+$~\cite{liu2024image}, which target advanced watermarking schemes~\cite{wen2023tree,yang2024gaussian}, incur substantial computational overhead due to repeated DDIM inversions. This reliance on iterative inversion not only introduces prohibitive latency that limits real-time deployment, but also accumulates reconstruction errors, resulting in noticeable visual degradation. Achieving a high-fidelity \textit{single-step} DDIM-based attack remains an open problem, as existing methods fail to learn a direct end-to-end restoration that can effectively decorrelate the invisible watermark from complex semantic content without iterative refinement.

\textbf{Contributions}. To address the above limitations, we propose \method, a model-agnostic watermark removal framework built on three key ideas: \ding{182} We introduce the Noise-Latent Alignment Score (NLAS) to measure and weaken the dependency between the generated latent representation and the initial noise, enabling watermark removal through constrained latent-space decorrelation while preserving image quality. \ding{183} We design \method for a \textit{no-box} setting, using only a public proxy model and no access to the target watermarking system, which supports transfer across different watermarking paradigms. \ding{184} We further develop \methodA, an amortized variant that distills the optimization objective into a Watermark Removal Network (WRN), replacing iterative optimization with a single feed-forward pass for more efficient removal. Specifically, the main contributions of this paper include:

\begin{itemize}
    \item We propose \method, a no-box, model-agnostic watermark removal attack based on NLAS. Across nine representative image watermarking schemes, \method reduces the average BA to 53.14\%, close to random guessing, while achieving the best overall visual quality among attack baselines. We further provide a theoretical analysis from the perspectives of flow matching and ODE theory to justify the removability of diffusion-model watermarks and the validity of NLAS.

    \item We develop \methodA, an amortized variant of \method that distills the latent manipulation objective into a WRN. \methodA removes the need for iterative optimization and enables efficient watermark removal, achieving 0.5 seconds per image with 6,282 MB of VRAM.
    
    \item We further evaluate \method and \methodA against SynthID-Image, achieving a 100\% ASR while maintaining superior perceptual quality, and extend evaluation to video watermarking schemes VideoShield and VideoMark, where \method suppresses BA to $\sim50\%$ and \methodA achieves excellent video quality among successful attack baselines (MUSIQ $>0.6$).
    
     \item We explore an attack detector as a defensive mechanism against \method and \methodA, demonstrating its potential to reliably identify compromised content. These findings further underscore the pressing need for watermarking schemes designed to be intrinsically robust against latent-space manipulation.

\end{itemize}

\section{Related Work}

\begin{table*}[t]
\centering
\caption{\textbf{Comparison of watermark removal attacks.} We categorize prior attacks by methodology and compare them in terms of proxy-model requirement and key capabilities: optimization-free execution, model-agnosticism, attack speed (seconds per image), fidelity preservation, adversary's knowledge and attackable modalities.}
\label{tab:related_work}
 \resizebox{\textwidth}{!}{%
\begin{tabular}{l l c c c c c c c c}
\toprule

&\multirow{2.5}{*}{\textbf{Method}} &\multirow{2.5}{*}{\textbf{Source}} &\multirow{2.5}{*}{\textbf{Proxy Model}}& \multicolumn{6}{c}{\textbf{Attack Capabilities}} \\
 \cmidrule(lr){5-10}
&   & & & \textbf{Opt.-free} & \textbf{Model-agn.} & \textbf{Speed} & \textbf{Fidelity} & \textbf{Knowledge} &\textbf{Modality} \\
\midrule

\multirow{3}{*}{\shortstack[l]{Regeneration}} 
 & VA \cite{an2024waves} &ICML'24& \cmark (VAE) & \cmark & \xmark & $0\sim1$  & \high & No box & Image\\
 & DA \cite{zhao2024invisible} &NeurIPS'24& \xmark (SDM) & \cmark & \xmark & $0\sim1$ & \high & No box & Image\\
  & NFPA \cite{qiufuture} &NeurIPS'25& \cmark (SDM) & \cmark & \xmark & $1\sim5$ & \low & No box & Image\\
\midrule

\multirow{6}{*}{\shortstack[l]{Adv. Attack}} 
 & Imprint \cite{muller2025black} &CVPR'25 & \cmark (SDM) & \xmark & \xmark & $10\sim50$ & \med & No box & Image\\
  & CtrlRegen+ \cite{liu2024image} &ICLR'25 & \cmark (SDM)& \xmark & \xmark & $5\sim10$ & \med & No box & Image\\
 & UnMarker \cite{kassis2025unmarker} &S\&P'25 & \xmark & \xmark  & \xmark & $>600$ & \med & No box & Image\\
 & Noiseprints Attack (NPA) \cite{goren2025noiseprints} &$-$ &\cmark (SDM)& \xmark & \xmark & $10\sim50$ & \med & No box & Image\\

\cmidrule(lr){2-10} 
 & \textbf{\method} &$-$ & \cmark (SDM) & \xmark & \cmark & $5\sim10$ & \textbf{High} & \textbf{No box} & \textbf{Image\&Video}\\
 & \textbf{\methodA} &$-$ & \cmark (SDM)& \cmark & \cmark & $0\sim1$ & \textbf{High} & \textbf{No box}  & \textbf{Image\&Video}\\
\bottomrule
\end{tabular}}
\end{table*}

\subsection{AI-generated Image Watermarking}
Digital watermarking has progressed from traditional pixel- and frequency-domain methods~\cite{van1994digital,al2007combined,navas2008dwt,wan2022comprehensive} to deep learning-based frameworks~\cite{zhu2018hidden,zhang2019robust}, serving as a key tool for provenance tracing and copyright protection. While conventional post-hoc methods watermark existing images, the rise of AIGC has shifted attention toward in-generation watermarking, which embeds signals during synthesis for improved stealthiness and robustness. Accordingly, T2I watermarking methods can be broadly categorized as follows:

\textbf{Fine-tuning-based watermarking:} Fine-tuning-based watermarking embeds watermarks by modifying generative model parameters, making the watermark an intrinsic property of the output distribution. Representative approaches operate on different components: Stable Signature~\cite{fernandez2023stable} fine-tunes the Variational Autoencoder (VAE) decoder to embed a fixed model-level watermark. Xiong et al.~\cite{xiong2023flexible} and WOUAF~\cite{kim2024wouaf} extend this to user-level attribution via learnable latent injection modules and weight modulation, respectively. SleeperMark~\cite{wang2024sleepermark} instead fine-tunes the U-Net with a trigger-based mechanism, enabling covert watermark activation conditioned on specific prompts.
    
\textbf{Initial-noise-based watermarking:} Initial-noise-based watermarking advances in-generation techniques by embedding watermarks directly into the earliest stage of the generation process. By exploiting the structure of Latent Diffusion Models (LDMs), this paradigm freezes model parameters and subtly modifies the initial Gaussian noise, allowing the watermark to propagate through the denoising trajectory and persist in the final output. TreeRing~\cite{wen2023tree} is a pioneering method that embeds a watermark signal guided by concentric ring patterns into the Fourier domain of the initial noise. Subsequent works enhance security by incorporating cryptographic primitives into the noise sampling process. Gaussian Shading~\cite{yang2024gaussian} partitions the Gaussian probability space to encode watermark bits, and T2SMark~\cite{yang2026t2smark} further improves key management and generation diversity. Similarly, Gunn et al.~\cite{gunnPRC} employed Pseudorandom Error-correcting Codes (PRC) to modulate noise signs via watermark-derived sequences, providing strong computational indistinguishability guarantees.

\subsection{Watermark Removal Attacks}
The widespread dissemination of generative content has intensified research on watermark removal attacks, which seek to undermine verification and attribution by manipulating watermarked images~\cite{zhao2024sok,cao2025secure}. While early attacks relied on simple distortions with limited effectiveness, modern adversaries employ sophisticated techniques capable of suppressing or eliminating watermarks while preserving perceptual fidelity. Existing removal strategies fall into three paradigms~\cite{an2024waves,kassis2025unmarker,qiufuture}: distortion attacks, regeneration attacks, and adversarial attacks, with the latter two representing the most effective approaches; representative methods are summarized in \Cref{tab:related_work}.

\textbf{Distortion attacks} are commonly used to assess the baseline robustness of watermarking schemes~\cite{fernandez2022watermarking,kim2024wouaf,wang2024sleepermark,wen2023tree,yang2024gaussian}. They apply standard signal-processing operations, such as geometric transformations and JPEG compression, to degrade watermark signals and reduce detection confidence. Although such attacks can weaken watermarks by introducing irreversible distortions or suppressing detection-relevant high-frequency components, their overall effectiveness is often limited.

\textbf{Regeneration attacks} aim to remove watermarks by reconstructing images using non-watermarked generative models. For example, Zhao et al.~\cite{zhao2024invisible} exploited the information bottleneck of VAEs and LDMs to filter out watermark signals during latent reconstruction. However, this approach relies on lossy encoding and decoding processes, which frequently introduce visual artifacts such as blurring and detail loss. Qiu et al.~\cite{qiufuture} proposed transforming watermarked images into semantically coherent \textit{future frames} using an unwatermarked generative model. By optimizing constrained latent warping and applying frame-attention mechanisms to preserve semantic content, their method successfully removes initial-noise-based watermarks in a black-box setting, albeit with potential fidelity trade-offs.

\textbf{Adversarial attacks} constitute one of the most extensively studied classes of watermark removal methods. These attacks optimize carefully crafted, imperceptible perturbations to mislead watermark detectors while preserving visual quality, typically formulated as constrained optimization problems. They are effective under varying adversarial knowledge assumptions and have demonstrated strong performance in bypassing detection systems. Hu et al.~\cite{hu2024transfer} studied the no-box watermark evasion via transferable adversarial perturbations generated from surrogate model ensembles. MarkSweep~\cite{cao2026marksweep} likewise targets the no-box setting, using a noise intensification–denoising pipeline to amplify and suppress imperceptible watermark signals. UnMarker~\cite{kassis2025unmarker} proposes a universal removal framework by optimizing Fourier-domain losses or applying adversarial filtering to alter structural image features.

Nevertheless, these adversarial approaches generally fail against initial-noise-based watermarking methods~\cite{wen2023tree,yang2024gaussian,gunnPRC,yang2026t2smark} and other highly robust watermarking schemes~\cite{wang2024sleepermark,lu2024robust}. Recent works have begun exploiting the generative priors of latent diffusion models for watermark removal under a no-box setting. Zhang et al.~\cite{lee2025removal} and M{\"{u}}ller et al.~\cite{muller2025black} optimized latent perturbations to cross detection boundaries with minimal visual distortion, while Liu et al.~\cite{liu2024image} guided controllable regeneration from noise using semantic and spatial constraints. Although these approaches differ in methodology, they collectively reveal a growing trend: leveraging the structure and dynamics of the diffusion process itself to suppress or remove embedded watermarks.

\section{Background}

\subsection{LDM-Based Image Generation}\label{sec:ldm}
We formulate our method within the framework of LDMs~\cite{rombach2022high}, a class of generative models that has become foundational in diffusion-based image synthesis. Specifically, an LDM is composed of two principal components: (i) A diffusion model $\mathcal{U}$ based on U-Net, which governs the stochastic denoising trajectory in latent space; and (ii) a VAE, typically decomposed into an encoder–decoder pair $(\mathcal{E},\mathcal{D})$, where $\mathcal{E}$ maps input images into compact latent representations, and $\mathcal{D}$ reconstructs the corresponding images from those latents. The full LDM is denoted as $\Theta = (\mathcal{E},\mathcal{U},\mathcal{D})$. The image generation process begins by sampling a noise vector $z_T$ from a standard Gaussian distribution. This latent is then progressively denoised to $z_0$, using the DDIM~\cite{song2020denoising} or DPM-Solver$++$~\cite{lu2025dpm,hong2024exact}, and finally decoded into the image space via the decoder $\mathcal{D}$:
\begin{equation}
   I = \mathcal{D} \left( \mathsf{DDIM}_{T \rightarrow 0}(z_T;\, \mathcal{U}) \right),\quad z_T \sim \mathcal{N}(\mathbf{0}, \mathbf{I}).\
\end{equation}

The inverse DDIM~\cite{mokady2023null,muller2025black} procedure reconstructs an image’s generative trajectory by running the deterministic reverse diffusion dynamics backward, which can be interpreted as a discretized ODE/flow trajectory under the probability-flow ODE interpretation~\cite{song2020denoising,songscore}. This yields an approximately invertible mapping between $z_0$ and $z_T$. At each timestep $t$, the latent is updated by:
\begin{equation}
z_{t+1} = \sqrt{\alpha_{t+1}}\, \hat{z_0}^t + \sqrt{1 - \alpha_{t+1}}\, \mathcal{U}(z_t, t, y),    
\end{equation}
where $\alpha_t$ denotes the noise schedule and $\hat{z_0}^t$ is the model’s estimate of the original latent. This formulation inverts the generative process without requiring access to the original condition $y$~\cite{mokady2023null}, thereby enabling reconstruction of the noise trajectory that could have produced $z_0$. This inversion underlies the basic principle of initial-noise-based watermarking methods~\cite{wen2023tree,yang2024gaussian,gunnPRC,yang2026t2smark}, where the watermark is embedded by modifying the $z_T$ and is then faithfully propagated through the generative process. We denote this inversion process as:
\begin{equation}
    z_{T} =\mathsf{DDIM}_{0\rightarrow T}(z_0;\, \mathcal{U}),\quad z_0=\mathcal{E}(I).
\end{equation}


 


\subsection{Formulation of Image Watermarking}

The embedding processes of \textbf{post-hoc watermarking} and \textbf{fine-tuning-based Watermarking} can be defined as superimposing a specific noise signal $\delta$ onto the original cover image $I$ to obtain the watermarked image $I_w$. This process can be formalized as:
\begin{equation}\label{eq:add_wm}
\begin{aligned}
    I_w = \Theta(z_T,y) + \delta, 
\end{aligned}    
\end{equation}
where $\delta$ denotes the watermark signal (or watermark residual). For post-hoc methods, $\delta$ is explicitly produced by the watermark embedding function; for fine-tuning-based methods, $\delta$ corresponds to the implicit output shift induced by parameter updates, i.e., $\Theta^w(z_T,y)-\Theta(z_T,y)$.

\textbf{Initial-noise-based watermarking} leverages the high sensitivity of diffusion models to initial conditions and the properties of DDIM, allowing the watermark embedded in the initial noise to propagate to the final generated image. This process can be defined as follows:
\begin{equation}
    I_w = \Theta(z_T^{w},y), \quad z_T^{w} \sim \mathcal{N}(\mathbf{0},\mathbf{I}),
\end{equation}
where $z_T^w$ is constructed in the latent space based on a pre-defined watermark $w$ using a specific embedding algorithm.

Existing watermarking methods generally fall into three paradigms, yet they share a common verification function. Formally, we define a unified detector $f:\mathcal{I}\to\mathbb{R}$ that consists of both watermark extraction and verification. For an input image, it outputs a scalar score compared against a threshold $\tau$ to determine watermark presence.

\section{Threat Model}

In image watermarking, \textbf{a watermark owner} (e.g., image creator, model provider, or copyright holder) employs a watermarking system to embed imperceptible signals into AI-generated images for provenance tracking and copyright enforcement, while \textbf{an adversary} seeks to evade verification by removing the watermark from the watermarked image while maintaining high visual fidelity of the image. We formalize watermark removal as an evasion attack under a \textit{no-box} setting, and detail the adversary’s knowledge and objectives.

    
\subsection{Adversary Knowledge}
The generative framework employed by the watermark owner typically relies on LDMs \cite{rombach2022high}, such as Stable Diffusion \cite{stabilityai_sd1.5,stabilityai_sd21base_2022,stabilityai_sdxl}, to ensure the synthesis of high-fidelity and semantically rich images. Depending on the specific provenance tracking strategy, the backbone model is either utilized as a fixed generator in post-hoc and initial-noise-based schemes or subjected to parameter modification in fine-tuning-based approaches. In this work, we consider a realistic and challenging \textit{model-agnostic} threat model, in which the adversary aims to remove potential watermarks without relying on any prior knowledge of the specific watermark embedding mechanisms. The adversary’s capabilities are defined as follows:
\begin{itemize}

\item \textbf{No-box Target Access:} The adversary can only generate watermarked images from the target generator $\Theta$, but has no access to its internal parameters, gradients, or specific architecture (e.g., the decoder or the extractor).

\item \textbf{No-Knowledge of Watermarking Scheme:} We assume a strictly no-knowledge setting regarding the target watermarking scheme. The adversary is unaware of the specific paradigm, embedding, extraction, and verification algorithms, or the embedded watermark $w$. Furthermore, we assume a \textit{no-oracle setting} in which the adversary has no access to the watermark detection interface and cannot query the verifier for scores or gradients.

\item \textbf{Data Constraints:} We assume that the adversary lacks access to ground-truth paired samples (i.e., a pair of watermarked and non-watermarked images depicting the same content or generated from the same prompt).

\item \textbf{Proxy Model Access:} The adversary has full white-box access to a publicly available proxy model $\Theta^{A}$, which is assumed to share a similar generative domain with the image generation model $\Theta$, but is clean.

\end{itemize}   

\subsection{Attack Objectives} 
The adversary aims to achieve effective watermark removal while maintaining high visual fidelity. We formalize these two competing objectives as follows:

\noindent \textbf{Watermark Evasion.} To compromise the target watermarking system, the attack goal is to suppress the detection score below the decision threshold, thereby causing a false negative:
\begin{equation}
f(I_A) < \tau,
\end{equation}
where $I_A$ denotes the adversarial image.

\noindent \textbf{Visual Stealthiness.} The attacked image $I_A$ after watermark removal must remain perceptually indistinguishable from the original watermarked image $I_w$ to ensure utility. We require:
\begin{equation}
d_{\text{img}}(I_A, I_w) \le \eta,
\end{equation}
where $d_{\text{img}}$ represents a perceptual distance metric (e.g., LPIPS~\cite{zhang2018unreasonable_LPIPS} or MSE), and $\eta$ is an acceptable distortion budget.

\noindent \textbf{Generality and Practicality.} 
The proposed attack is model-agnostic and designed to generalize across diverse watermarking paradigms and real-world deployment scenarios. Unlike prior approaches that rely on exploiting scheme-specific embedding artifacts, our attack leverages universal properties of LDMs, enabling consistent effectiveness against post-hoc, fine-tuning-based, and initial-noise-based watermarking schemes. Moreover, the attack assumes no access to the target model or watermarking mechanism and must successfully transfer from a publicly available proxy model to unknown or proprietary systems without modification.
Meanwhile, we emphasize practicality and stealth. The attack is required to operate under realistic resource constraints, achieving low latency and modest memory overhead suitable for consumer-grade hardware, while preserving high perceptual quality to avoid raising suspicion during watermark removal.

\section{Proposed Attacks}

\begin{figure}[t]
    \centering
    \includegraphics[width=\linewidth]{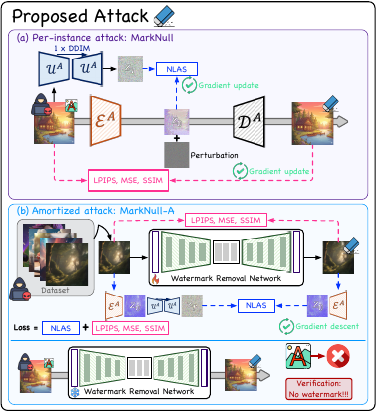}
    \caption{Illustration of the attack pipeline. \method iteratively optimizes latent representations to evade detection; \methodA trains a WRN for rapid end-to-end removal. Both modes leverage a proxy model $\Theta^A$.}
    \label{fig:overview}
\end{figure}

\subsection{Overview}
We propose \method and its end-to-end version \methodA, as illustrated in~\Cref{fig:overview}. \method is a per-instance optimization attack that removes watermark through latent space manipulation. It iteratively optimizes a perturbation applied to the image's latent representation $z_0$, balancing two competing objectives: (i) maximizing deviation from the watermarked initial noise distribution to erase embedded watermark signals, and (ii) minimizing perceptual loss to preserve visual fidelity. To enhance attack efficiency, we further introduce \methodA, an amortized version of the attack designed for rapid inference. \methodA trains a WRN to learn a generalized mapping from watermarked image to watermark-free image in the pixel domain, distilling the iterative optimization objective into a single forward pass, while maintaining the hybrid loss constraints on both latent and pixel spaces, achieving a favorable balance between removal effectiveness and image quality.

 \subsection{Key Insights}\label{sec:insights}
  As illustrated in \Cref{sec:ldm}, the watermarked $z_T$ can be transported along the denoising trajectory and manifest in $z_0$ and thus in the decoded image, making $z_T$ a natural carrier for watermark information. Our proposed approach is driven by the fundamental observation that watermarking paradigms, regardless of their implementation, rely on preserving specific statistical dependencies along the generative trajectory ($z_T \to z_0$). We introduce Latent Space Manipulation to disrupt these dependencies through two complementary mechanisms:

  \textbf{Statistical Independence via Orthogonality.} For initial-noise-based watermarking schemes, we observe that simple distance maximization in latent space is insufficient to remove watermark signals. Instead, we enforce geodesic orthogonality between the optimized latent $z_0^\mathrm{adv}$ and the original initial noise $z^{A,w}_T$. In high-dimensional latent spaces, orthogonality implies statistical independence, severing the causal dependency between the generated image and the watermarked initial noise. As a result, the re-inverted noise is mapped to a pristine, unstructured distribution rather than a watermark-encoded pattern, enabling reliable watermark removal.
  
  \textbf{Manifold Projection and VAE Filtering.} For fine-tuning and post-hoc watermarking schemes, our attack performs regeneration. First, the VAE information bottleneck functions as a low-pass filter, suppressing high-frequency, pixel-level artifacts commonly introduced by post-hoc watermark embedding. Second, by optimizing $z_0^\mathrm{adv}$ to minimize perceptual distortion while simultaneously enforcing the orthogonality constraint, we induce a controlled distributional shift: the resulting adversarial image maintains high visual fidelity but is projected away from the model’s specific \textit{watermarked} statistical manifold. This displacement effectively prevents watermark detection while preserving semantic content.

\subsection{\method}

The initial noise $z_T$ leaves a persistent structural trace on the final generation; specifically, the latent representation $z_0$ retains a significantly higher correlation with its originating source $z_T$ compared to independent random noise. Moreover, the $z_T\rightarrow z_0$ trajectory is constrained by the denoiser and the scheduler. The generated states do not densely cover the ambient Euclidean space; instead, they are concentrated near a structured low-dimensional manifold. Thus, we introduce NLAS to measure the correlation between $z_0$ and $z_T$:
\begin{equation}
\mathrm{NLAS}(z_0,z_T) = ( \arccos\left( \frac{z_0 \cdot z_T}{\|z_0\|_2 \|z_T\|_2} \right) - \frac{\pi}{2} )^2.
\end{equation}

\Cref{sec:nlas} validates NLAS as an effective correlation metric: $z_0$ shows high alignment with its coupled $z_T$, whereas it is effectively orthogonal to random standard Gaussian noise $z_{\text{rand}} \sim \mathcal{N}(\mathbf{0}, \mathbf{I})$.

Therefore, we aim to construct an adversarial latent $z_0^{(\mathrm{adv})}$ that is orthogonal to $z_T$, formulated as follows:
\begin{equation}
\begin{aligned}
        &\theta(z_0^{(\mathrm{adv})}, z_T) \to \frac{\pi}{2} \implies \left| \langle \frac{z_0^{(\mathrm{adv})}}{\|z_0^{(\mathrm{adv})}\|}, \frac{z_T}{\|z_T\|} \rangle \right| \approx 0,\\
    \end{aligned}
\end{equation}
where $\langle \cdot,\cdot \rangle$ denotes the inner product. To promote statistical independence under a Gaussian prior, we enforce orthogonality between $z_0$ and $z_T$.

To preserve visual fidelity while optimizing the adversarial $z_0^{(\mathrm{adv})}$, we incorporate an image quality loss term that constrains the decoded adversarial image to remain close to the original watermarked image. Specifically, we define the quality loss as:
\begin{equation}
\begin{aligned}
    \mathcal{L}_{\mathrm{img}}(I_A, I_w)
&= \lambda_1\,\mathcal{L}_\mathrm{LPIPS}(I_A, I_w)
+ \lambda_2\,\|I_A - I_w\|_2^2\\
&+ \lambda_3\,\big(1-\mathcal{L}_\mathrm{SSIM}(I_A, I_w)\big),
\end{aligned}
\end{equation}
where $I_w$ and $I_A$ denote the original watermarked image and the image decoded from the adversarial latent via proxy model, respectively; $\mathcal{L}_\mathrm{LPIPS}$ measures perceptual dissimilarity, $\|\cdot\|_2^2$ is the pixel-wise Mean Squared Error (MSE), and $\mathcal{L}_\mathrm{SSIM}$ quantifies structural similarity. The weights $\{\lambda_1,\lambda_2,\lambda_3\}$ balance the perceptual, pixel-level, and structural fidelity terms. This loss ensures that watermark removal does not produce noticeable image degradation.

Consequently, we aggregate the aforementioned terms into a unified objective function and optimize the adversarial $z_0^{(\mathrm{adv})}$ via gradient descent. Crucially, this optimization process relies exclusively on the proxy model $\Theta^A=(\mathcal{E}^A, \mathcal{U}^A,\mathcal{D}^A)$, adhering strictly to the no-box setting. The optimization procedure is summarized in \Cref{alg:MarkNull}. We name this per-instance attack \method.
\begin{equation}\label{eq:loss1}
\begin{aligned}
   \mathcal{L}_{\method} =  \mathcal{L}_{\mathrm{img}}(I_A,I_w) + \lambda_4\mathrm{NLAS}(z_0^{(\mathrm{adv})},z_T^{A,w}).
\end{aligned}
\end{equation}

\begin{algorithm}[t]
    \caption{\method}
    \label{alg:MarkNull}
    \SetKwInOut{Input}{Input}
    \SetKwInOut{Output}{Output}
    \Input{
        Watermarked image $I_w$, 
        proxy model $\Theta^A$, 
        latent perturbation budget $\epsilon$; 
        steps $N$, learning rate $\kappa$.
    }
    \Output{Watermark-free image $I_A$.}
    \textcolor{blue}{\tcp{1. Initialization}}
    \DontPrintSemicolon
    $z_0^{A,w} \leftarrow \mathcal{E}^A(I_w)$\;
    $z_T^{A,w} \leftarrow \mathsf{DDIM}_{0\rightarrow T}(z_0^{A,w}; \mathcal{U}^A)$\;
     $z_0^{(\mathrm{adv})} \leftarrow z_0^{A,w}$\;
    
    \textcolor{blue}{\tcp{2. Optimization Loop}}

    \For{$i \leftarrow 1$ \KwTo $N$}{
    $I_A^{(i)} \leftarrow \mathcal{D}^A(z_0^{(\mathrm{adv})})$\;

    $g_i \leftarrow \nabla_{z_0^{(\mathrm{adv})}}
    \mathcal{L}_{\method}(z_0^{\rm (adv)},z_T^A,I_A^{(i)}, I_w)$
    \bluetri{\Cref{eq:loss1}}\;

    $z_0^{(\mathrm{adv})} \leftarrow
    z_0^{(\mathrm{adv})} - \kappa g_i$\;

    $z_0^{(\mathrm{adv})} \leftarrow
    z_0^{A,w} +
    \mathsf{Clip}(
        z_0^{(\mathrm{adv})}-z_0^{A,w},
        -\epsilon,\epsilon)$\;}
    
    \textcolor{blue}{\tcp{3. Final Reconstruction}}
    $I_A \leftarrow \mathcal{D}^A(z_0^{(\mathrm{adv})})$\;
    \KwRet{$I_A$}
\end{algorithm}

Furthermore, in \Cref{sec:att_proof}, we provide a theoretical analysis of the feasibility and generality of \method. We view the generation dynamics of LDMs as a deterministic ODE flow driven by a learned time-dependent vector field, which allows us to leverage flow matching and ODE theory to establish the existence of an inverse mapping and thus the removability of watermarks across a broad range of watermarking schemes.

\subsection{\methodA}
While \method serves as a powerful per-instance attack, its reliance on iterative gradient-based optimization introduces substantial computational latency, limiting scalability in real-time or large-scale deployment scenarios. However, as empirically demonstrated, the perturbations learned by \method are not purely random; they exhibit consistent and transferable structural patterns across different samples.


This observation suggests that \textit{\textbf{the transformation required to decorrelate watermark signals from the latent diffusion trajectory is both learnable and generalizable}}. Motivated by this insight, we introduce \methodA, an amortized variant designed to improve attack efficiency. Rather than performing iterative optimization on individual samples, \methodA employs a WRN to amortize the optimization cost through a single feed-forward model. The WRN is trained to learn a direct mapping from watermarked images to their watermark-free counterparts by minimizing the following objective function over a representative training dataset.
\begin{equation}\label{eq:loss2}
\begin{aligned}
    \mathcal{L}_{\methodA} = \mathbb{E}_{I \sim \mathcal{D}} \Big[
    \lambda_5\,\mathcal{L}_{\mathrm{LPIPS}}(I_A, I)
    + \lambda_6\,\| I_A - I \|_2^2& \\
    \quad + \lambda_7\,\big(1-\mathcal{L}_\mathrm{SSIM}(I_A, I)\big) +\lambda_8\,\mathrm{NLAS}(z_0^{(\mathrm{adv})}, z^A_T) \Big].
\end{aligned}
\end{equation}

\subsubsection{WRN}
We build the WRN upon a Restormer~\cite{zamir2022restormer}, a SOTA Transformer architecture designed for high-fidelity image restoration. The Restormer is well-suited to our task for two main reasons: (i) its multi-head transposed self-attention mechanism effectively captures long-range dependencies while maintaining computational efficiency for high-resolution inputs; and (ii) its restoration-focused inductive bias, an encoder–decoder hierarchy with multi-scale feature refinement, closely aligns with our goal of removing the watermark without introducing perceptible distortions. 

\begin{algorithm}[t]
\caption{\methodA}
\label{alg:MarkNull-A}
\DontPrintSemicolon
\KwIn{Dataset $\mathcal D$, proxy model $\Theta^A$, learning rate $\xi$, image $I_w$}
\KwOut{$\mathcal W_\phi$, watermark-free image $I_A$}

Initialize $\phi$\;
\While{not converged}{
    Sample $B\sim\mathcal D$\;
    \textcolor{blue}{\tcp{1. Forward}}
    $z_T^A\leftarrow\mathsf{DDIM}_{0\to T}
    (\mathcal E^A(B);\mathcal U^A)$\;
    $\widetilde B\leftarrow\mathcal W_\phi(B)$\;
    $z_0^{\rm(adv)}\leftarrow\mathcal E^A(\widetilde B)$\;

    \textcolor{blue}{\tcp{2. Gradient Update}}
    $\phi\leftarrow\phi-\xi\nabla_\phi
    \mathcal L_{\methodA}(z_0^{\rm (adv)},z_T^A,\widetilde B,B)$
    \bluetri{\Cref{eq:loss2}}\;
}
\textcolor{blue}{\tcp{3. End-to-end Attack}}
$I_A\leftarrow
\mathcal D^A(\mathcal E^A(\mathcal W_\phi(I_w)))$\;
\KwRet{$\mathcal W_\phi,I_A$}\;
\end{algorithm}

    
    
    
        

To comply with the \textit{Data Constraints}, we construct the training dataset $\mathcal{D}$ utilizing a publicly available proxy model $\Theta^A$. $\mathcal{D}$ consists exclusively of clean AI-generated images $I$ without any embedded watermarks. Accordingly, the WRN $\mathcal{W}_\phi$ is trained to produce an output image $I_A$ from $I$ that preserves high perceptual fidelity while enforcing latent-space decorrelation and orthogonality between the optimized latent representation $z_0^{A}$ and the initial noise $z^A_T$. The training procedure is summarized in \Cref{alg:MarkNull-A}. To further enhance attack efficacy and ensure consistency with the diffusion prior, the output of $\mathcal{W}_\phi$ is subsequently passed through a VAE regeneration using $\Theta^A$.

\section{Provable Attack Guarantees of \method}\label{sec:att_proof}

In this section, we theoretically substantiate the effectiveness of \method. The generation dynamics of an LDM can be viewed as a deterministic ODE flow driven by a learned time-dependent vector field \cite{holderrieth2025introduction}. This perspective enables us to leverage flow matching and ODE theory to formalize the existence of an inverse mapping, thereby establishing the removability of watermarks. Our argument is structured around three key observations: (i) because the generative flow is a diffeomorphism, the induced watermark map admits a unique preimage, ensuring well-posed invertibility; (ii) convergence of the marginal vector fields implies that discrete inversion procedures are asymptotically consistent across model variants; and (iii) a latent-space decorrelation mechanism suppresses alignment-based test statistics, reducing detectability under common alignment metrics.

To be consistent with the LDM literature \cite{rombach2022high}, we use $z_T$ to denote the initial noise and $z_0$ the generated latent representation. We reparameterize time $t \in [0,1]$ such that $X_0 \sim p_{\text{init}}$ corresponds to $z_T$, and $X_1 \sim p_{\text{data}}$ corresponds to $z_0$ following ODE and flow matching.

\subsection{Bijective Structure of Noise-Data Mapping}
In this section, we establish the theoretical validity of altering the initial noise by manipulating the latent representation.

\begin{assumption}[\textbf{Deterministic Marginal Flow Analysis}]~\label{assump:dete_flow} 
We define the generative process as a deterministic ODE flow induced by the \textit{Marginal Vector Field}. Thus, the trajectory satisfies:
\begin{equation}
    \frac{dX_t}{dt} = u_t^{\text{target}}(X_t), \quad X_0 = z_T \implies X_1 = z_0,
    \label{eq:ode_flow}
\end{equation}
where $u_t$ denotes the vector field.
\end{assumption}

This assumption is grounded in the theoretical framework in \cite{holderrieth2025introduction}. This validates the marginalization trick, demonstrating that a marginal vector field governing the target probability path can be constructed by integrating conditional vector fields over the posterior distribution. This establishes the theoretical basis for treating the generation process as a deterministic ODE flow; consequently, deterministic samplers used in practice (e.g., DDIM \cite{song2020denoising}, DPM-Solver \cite{lu2025dpm}) are essentially numerical discretizations of this specific ODE.

\begin{proposition}[\textbf{Bijectivity Implies a Unique Preimage}]
Let $\psi: K \to \psi(K)$ be the time-1 flow map induced by an ODE on the compact set $K$, satisfying $z_0 = \psi(z_T)$. If the vector field $u_t$ is continuously differentiable and uniformly Lipschitz on $K$, then $\psi$ is a diffeomorphism. For any generated latent representation $z_0 \in \psi(K)$, there exists a unique preimage $z_T = \psi^{-1}(z_0)$. Therefore, any non-zero perturbation $\epsilon$ resulting in $z^{\mathrm{(adv)}}_0 = z_0 + \epsilon$ (provided $z^{\mathrm{(adv)}}_0 \in \psi(K)$) necessarily corresponds to a shifted initial state $z^{\mathrm{(adv)}}_T \neq z_T$.

\begin{proof}
The solution to the ODE satisfying the aforementioned regularity conditions defines a diffeomorphism. Injectivity guarantees the uniqueness of the preimage. Furthermore, since the inverse map $\psi^{-1}$ is smooth on the compact image set $\psi(K)$, its Jacobian norm is bounded. According to the Mean Value Theorem, the variation in the initial noise is bounded by:
\begin{equation}
    \|z^{\mathrm{(adv)}}_T  - z_T\| \le L_{\text{inv}}\|\epsilon\|,
\end{equation}
where $\epsilon$ denotes the adversarial perturbation added to $z_0$, and $L_{\text{inv}}$ is the Lipschitz constant of the inverse mapping $\psi^{-1}$. This verifies that perturbations in the latent representation can effectively alter the original watermarked $z_T$.
\end{proof}
\end{proposition}

\subsection{Consistency of Watermark Estimation}\label{sec:consistency}
\method leverages a proxy model $\Theta^{A}$ to infer the watermarked initial noise. A key requirement is that the inversion results produced by the proxy model align with those of the target model.

\begin{assumption}[\textbf{Proxy-Target Compatibility}]\label{assump:pro-tar-comp} 
We assume that the target model $\Theta$ and the proxy model $\Theta^A$ share the same latent-space parameterization and the same discrete inversion scheme $\hat\psi^{-1}$ (i.e., the same noise schedule and solver type). In practice, $\hat\psi$ is $\psi$'s discrete approximation. Moreover, we assume that their vector fields $u_t^\Theta$ and $u_t^{\Theta^A}$ satisfy a uniform Lipschitz condition on a compact set $K$, and that all intermediate states encountered during inversion remain within $K$. For watermarking methods built upon SDM-based T2I models, implementing initial-noise-based watermarking typically relies on a largely standardized choice of noise schedule and solver type, which is widely accepted in the community and does not violate the no-box setting.
\end{assumption}

This assumption is grounded in the fundamental objective of generative modeling. Although $\Theta$ and $\Theta^A$ may differ in architecture or training details, both are optimized to approximate the same underlying ground-truth distribution of natural images, denoted as $p_{data}$. Consequently, they learn to map Gaussian noise to the same low-dimensional natural image manifold. Since the watermark removal process in \method essentially corresponds to a projection operation that pulls the sample back onto this shared manifold (filtering out off-manifold watermark perturbations), the gradient direction learned from the proxy manifold $\mathcal{M}(\Theta^A)$ serves as a highly effective approximation for the target manifold $\mathcal{M}(\Theta)$. This structural alignment ensures attack transferability even across disparate model architectures.

\begin{lemma}\label{lemma}
Under \Cref{assump:pro-tar-comp} and the stability theory of numerical integration (specifically, the discrete Grönwall lemma), the discrepancy in the inversion results is bounded by the error of the vector fields:
\begin{equation}
\sup_{z \in \psi(K)} \| \hat{\psi}^{-1}_{\Theta}(z) - \hat{\psi}^{-1}_{\Theta^A}(z) \| \leq C_{\text{stab}} \sup_{x \in K, t \in [0,1]} \| u^{\Theta}_t(x) - u^{\Theta^A}_t(x) \|.    
\end{equation}

This bound holds for any single-step or multi-step solver satisfying zero-stability and consistency conditions. The constant $C_{\text{stab}}$ depends on the Lipschitz constant, the number of steps, and the time interval.
\end{lemma}

\begin{proposition}[\textbf{Guarantee of Attack Transferability}]
Let $z_T^{A, w} = \hat{\psi}^{-1}_{\Theta^A}(z_0^{A, w})$ denote the proxy initial noise used in \Cref{alg:MarkNull} and \Cref{alg:MarkNull-A}, and let $z_T^{w} = \hat{\psi}^{-1}_{\Theta}(z_0^{w})$ denote the initial noise under the target model. If both models are sufficiently aligned such that $\sup_t | u_t^\Theta - u_t^{\Theta^A} | \leq \varepsilon_u$, then optimizing the adversarial latent representation $z_0^{(\text{adv})}$ to decorrelate from $z_T^{A, w}$ implicitly reduces its correlation with $z_T^{w}$.

\begin{proof}
The flow matching training objective regresses $u_t$ to a unique marginal vector field $u^{\text{target}}_t$ \cite{holderrieth2025introduction}. Assuming sufficient convergence, the vector field discrepancy $\varepsilon_u$ is minimal. Applying \Cref{lemma}, the inversion error is bounded by:
\begin{equation}
    \| z^{w}_T - z^{A, w}_T \| \le C_{\text{stab}} \varepsilon_u. 
\end{equation}
Therefore, \method minimizes $\mathrm{NLAS}(z_0^{(\mathrm{adv})}, z_T^{A, w})$, which in turn approximates minimizing the absolute inner product $|\langle z_0^{(\mathrm{adv})}, z_T^{w} \rangle|$. By applying the triangle inequality and Cauchy-Schwarz inequality to the decomposition, we obtain:
\begin{equation}
\begin{aligned}
 &|\langle z_0^{(\mathrm{adv})}, z_T^{w} \rangle|= |\langle z_0^{(\mathrm{adv})}, z_T^{A, w} + (z_T^{w} - z_T^{A, w}) \rangle| \\
&\le |\langle z_0^{(\mathrm{adv})}, z_T^{A, w} \rangle| + |\langle z_0^{(\mathrm{adv})}, z_T^{w} - z_T^{A, w} \rangle| \quad  \\
&\le \underbrace{|\langle z_0^{(\mathrm{adv})}, z_T^{A, w} \rangle|}_{\text{NLAS}} + \underbrace{\|z_0^{(\mathrm{adv})}\| \cdot \|z_T^{w} - z_T^{A, w}\|}_{C_{\text{stab}} \varepsilon_u }.
\end{aligned}   
\end{equation}

When normalized $\|z_0^{(\mathrm{adv})}\|=1$, and by \Cref{lemma} the error term is upper-bounded by $C_{\mathrm{stab}}\varepsilon_u$. Therefore, when the optimization drives the correlation (i.e., NLAS) under the proxy model toward $0$, the correlation under the target model is guaranteed to remain within $C_{\mathrm{stab}}\varepsilon_u$. As long as the two models are similar (i.e., $\varepsilon_u$ is small), the attack can be successfully realized using the proxy model.
\end{proof}
\end{proposition}

\subsection{Watermark Removal Guarantee}

Building upon the transferability guarantee established in the previous section, we now formally analyze the mechanism by which \method removes the watermark signal. While Section \ref{sec:consistency} ensures that our proxy-guided optimization remains valid for the target model, this section demonstrates that the optimization objective itself, minimizing NLAS, mathematically leads to a reduction of correlation-based detection statistics.

\begin{proposition}[\textbf{Watermark Removal via Alignment Suppression}]
This proposition applies to initial-noise-based watermarking schemes that rely on correlation-based detection such as TreeRing~\cite{wen2023tree}, Gaussian Shading~\cite{yang2024gaussian}, and PRC~\cite{gunnPRC}, where the decision statistic depends on the alignment between the recovered watermark $z_T^w$ and the original initial noise. By minimizing $\mathrm{NLAS}$, \method reduces the correlation between the adversarial latent $z_0^{(\mathrm{adv})}$ and the target watermark $z_T$ to the statistical baseline of independent random vectors.

\begin{proof}
\method explicitly minimizes the NLAS. Since NLAS is defined as a monotonic metric quantifying the deviation of the included angle from $\pi/2$, the condition $\text{NLAS} \to 0$ implies $\langle z_0^{(\text{adv})}, z_T \rangle \to 0$ (i.e., orthogonality). In high-dimensional spaces, according to the \textit{\textbf{Concentration of Measure phenomenon}}, independent isotropic random vectors are approximately orthogonal with extremely high probability (correlation concentrates around 0). Consequently, this optimization process suppresses the alignment metric to the statistical baseline of "random noise pairs". For detectors relying on this specific alignment structure, the detection statistic will degenerate to the Null Hypothesis $\mathcal{H}_0$, classifying the image as non-watermarked, leading to detection evasion.
\end{proof}
\end{proposition}

The analysis above primarily addresses watermarks rooted in the generation dynamics. However, watermarking schemes based on fine-tuning and post-hoc methods operate through different mechanisms. These approaches typically introduce perturbations in the high-frequency domain or alter the statistical distribution of the model. In the following proposition, we extend our theoretical framework to demonstrate that our method remains effective against these paradigms by leveraging the manifold projection property of the VAE decoder.

\begin{proposition}[\textbf{Generalization to Fine-tuning and Post-hoc Schemes}]
The adversarial latent perturbation generated by \method is effective against fine-tuning and post-hoc watermarking schemes. Since these watermarks can be modeled as additive perturbations $\delta$ that typically reside in the high-frequency space of the image manifold, \method evades detection by projecting the watermarked image onto the tangent space of the generative manifold, thereby filtering out the watermark signal.
\end{proposition}

\begin{proof}
Let $\mathcal{M} = \{ \mathcal{D}^A(z_0) \mid z_0 \in \mathcal{Z} \}$ be the image manifold. Post-hoc and fine-tuning watermarks are modeled as $I_w = I_{\text{clean}} + \delta$. \method optimizes the latent representation to minimize the reconstruction loss, which is mathematically equivalent to solving a local projection problem:
\begin{equation}
    \min_{\Delta z} \| \mathcal{D}^A(z_0 + \Delta z) - (I_{\text{clean}} + \delta) \|_2^2.
\end{equation}
By performing a first-order Taylor expansion of the decoder around $z_{\text{clean}}$ with Jacobian $J_z$, the optimization simplifies to a linear least-squares problem: $\min_{\Delta z} \| J_z \Delta z - \delta \|_2^2$. The closed-form solution yields the update in pixel space:
\begin{equation}
    \Delta I_A \approx J_z \Delta z^* = P_T \delta, \quad \text{where } P_T = J_z (J_z^\top J_z)^{-1} J_z^\top.
\end{equation}
Here, $P_T$ is the orthogonal projection matrix onto the tangent space $T_{I_{\text{clean}}}\mathcal{M}$. The reconstruction $I_A$ effectively retains the component $P_T \delta$ (on-manifold content) while discarding $(I - P_T)\delta$ (off-manifold noise).
Since fine-tuning and post-hoc watermarks are designed to be imperceptible, they predominantly occupy the high-frequency subspace orthogonal to the semantic manifold (i.e., $\| (I - P_T)\delta \| \gg \| P_T \delta \|$). Consequently, the projection $P_T$ significantly attenuates the watermark energy, rendering it undetectable.
\end{proof}

\section{Evaluation and Analysis}

We conduct comprehensive evaluations of the proposed \method and \methodA across a broad range of tasks, evaluation metrics, and representative baselines. \Cref{sec:exp_set} outlines our experimental setup, including the used datasets and evaluation metrics. In \Cref{sec:exp_res,sec:synthID,sec:video}, we present four key aspects of our analysis: attack effectiveness, image quality preservation, efficiency and usability, and generalization. Together, these experiments demonstrate the capability of our proposed attacks to address key research gaps in watermark removal.

\subsection{Experimental Setup}\label{sec:exp_set}

\method and \methodA are implemented with the hyperparameter settings described in \Cref{sec:marknull2_set}, where we also report the corresponding experimental results.

\textbf{Watermarking Baselines.} We deploy nine representative AI-generated image watermarking baselines, covering three major paradigms. Post-hoc watermarking methods include DWT-DCT-SVD (DwtDct) \cite{navas2008dwt}, RivaGan \cite{zhang2019robust}, and HiDDeN \cite{zhu2018hidden}. Fine-tuning-based methods include Stable Signature (SS) \cite{fernandez2023stable}, VINE \cite{lu2024robust}, and SleeperMark \cite{wang2024sleepermark}, which embed watermarks by updating the generative model parameters. Initial-noise-based methods include Gaussian Shading (GS) \cite{yang2024gaussian}, PRC \cite{gunnPRC}, and TreeRing (TR) \cite{wen2023tree}, which encode watermarks by manipulating the initial noise in the diffusion process. Among these, TR is a zero-bit watermarking scheme that verifies watermark presence without decoding a bit-string watermark.


\textbf{Model and Dataset.} We leverage Stable Diffusion 2.1-base (SD2.1)~\cite{stabilityai_sd21base_2022} as the default backbone for both image generation and watermark embedding. SD2.1 is a widely adopted open-source T2I model known for its high-fidelity outputs, and it serves as an underlying model for constructing AI-generated image watermarking schemes. To benchmark watermarking performance, we randomly sample prompts from the Stable Diffusion Prompts (SDP) dataset~\cite{gustavosta_stable_diffusion_prompts} and generate 100 watermarked images of resolution $512 \times 512$ for each watermarking method using SD2.1. All images are generated with a fixed number of 50 inference steps to ensure consistency across comparisons. For post-hoc watermarking baselines, we instead generate 100 images using Stable Diffusion XL-base 1.0 (SDXL1.0)~\cite{stabilityai_sdxl}, followed by the application of watermarking schemes after image synthesis. In addition, when a proxy model is required by the adversary within attack methods, we instantiate it using Stable Diffusion 1.5 (SD1.5). 


\textbf{Attack Baselines.} To comprehensively assess our attacks' capability for image watermark removal, we evaluate \method and \methodA against a broad suite of attack baselines covering both standard distortions and advanced watermark removal attacks. Specifically, five standard distortions include Gaussian noise, Gaussian blur, contrast adjustment, brightness adjustment, and JPEG compression. The advanced watermark removal attacks are further grouped into: (i) regeneration attacks, including Diffusion-Attack (DA) \cite{zhao2024invisible}, VAE-Attack (VA) \cite{an2024waves}, CtrlRegen+ \cite{liu2024image}, and Next Frame Prediction Attack (NFPA) \cite{qiufuture}; and (ii) adversarial attacks, including UnMarker \cite{kassis2025unmarker}, Imprint \cite{muller2025black}, and NoisePrints (NPA) \cite{goren2025noiseprints}. For each baseline, we use the attack hyperparameters reported in the original papers, so that the resulting outputs exhibit comparable visual quality, thereby enabling a fair, controlled comparison of watermark removal effectiveness. 


\textbf{Evaluation Metrics.} We employ five metrics to evaluate both reconstruction fidelity and perceptual realism, including PSNR$\blacktriangle$, SSIM$\blacktriangle$, LPIPS$\blacktriangledown$, FID$\blacktriangledown$, and BRISQUE$\blacktriangledown$. BRISQUE~\cite{mittal2012no}, a no-reference metric, evaluates perceptual quality without the ground truth. To facilitate a holistic assessment of image quality, we introduce a CQS$\blacktriangle$. In addition, we evaluate an attack's ability to evade watermark detection using BA and TPR@1\%FPR. Detailed definitions of these metrics are provided in \Cref{sec:metric_detail}.

\begin{table*}[t]
\centering
\small
\setlength{\tabcolsep}{4pt}
\renewcommand{\arraystretch}{0.9}
\caption{Attack evaluation against watermarking schemes (TPR@1\%FPR$\blacktriangledown$ for zero-bit watermarking; BA[\%] $\blacktriangledown$ for multi-bit watermarking). }
\label{tab:attack_vs_methods}
\begin{tabular}{p{2.1cm}*{9}{P{1.2cm}}|c}
\toprule
& \multicolumn{9}{c|}{\textbf{Multi-bit}}  & \textbf{Zero-bit}  \\
\midrule

\textbf{Attack}  & DwtDct & RivaGan & HiDDeN & SS & VINE &SleeperMark & GS & PRC &\textbf{Avg.} & TR\\
\midrule
\textit{No Attack}      &100  &99.96  &98.99  &99.76  &99.79  &99.96  &100 &100 &99.81 &100 \\
\midrule
Noise     &84.81  &99.81  &53.91  &96.85 &99.45  &99.31  &99.75  &89.94  &90.48  &92   \\
Blur  &100  &99.96  &72.12  &85.08  &99.81  &99.98  &100  &100 &94.62  &92   \\
Contrast   &50.64  &99.40  &87.12  &98.16  &99.71  &99.98  &100  & 100 &91.88  &98    \\
JPEG            &99.66  &98.62  &77.85  &79.77  &99.74  &99.75  &99.93  &95.82 &93.89  &91   \\
Brightness &50.19 &99.44 &89.63 &98.40 &99.70 &99.94 &100&100 &92.16 &96 \\
\midrule
VA \cite{zhao2024invisible}  &78.97  &64.78  &60.81  &62.66  &96.68  &95.98  &98.59  &84.47  &80.17  &83 \\
DA  \cite{an2024waves}    &69.64  &61.75  &61.87  &48.80  &91.67  &99.81  &99.74  &96.76    &78.76 &78  \\
CtrlRegen$+$ \cite{liu2024image}  &53.16  &57.90  &59.74  &47.64  &68.96  &89.10  &91.92  &52.18  & 65.08  &31 \\
NFPA \cite{qiufuture}  &52.93  &62.28  &58.37  &48.36  &52.76  &85.27 &50.60  &49.41 &57.50   &0 \\
Imprint \cite{muller2025black} & 59.81  &54.34  &61.22 &46.52 &73.70  &83.71  &31.68  &47.94 & 57.37 &34    \\
UnMarker \cite{kassis2025unmarker} &100 &87.84&63.91&93.41&98.59 &99.77 &99.48 &91.47 &91.81 &46 \\
NPA  \cite{goren2025noiseprints}   &56.91 &58.59 &60.29 &46.73 &84.16 &92.75 &66.77 &51.76&64.75  &29  \\

\midrule
\rowcolor{green!7}
\method &55.10 &59.78 &60.62 &49.12 &64.93 &42.88 &41.55 &51.18 &\textbf{53.14} &43\\
\rowcolor{green!7}
\methodA &66.43  &70.53 &59.85 &46.52 &56.00 &39.29  &70.28 &47.94 &\textbf{57.10} &63\\
\bottomrule
\end{tabular}
\end{table*}

\begin{table*}[t]
\centering
\small
\setlength{\tabcolsep}{3.6pt}
\renewcommand{\arraystretch}{0.9}
\caption{CQS evaluation of various attack methods against watermarking schemes.}
\label{tab:CQS_vs_methods}
\begin{tabular}{p{2.1cm}*{10}{P{1.2cm}}}
\toprule

\textbf{Attack}  & DwtDct & RivaGan & HiDDeN & SS & VINE &SleeperMark & GS & PRC  & TR & \textbf{Avg.}\\
\midrule
VA & 3.19 & 3.30 & 3.16 & 3.52 & 3.74 & 3.39 & 3.63 & 3.34 & 3.50 &3.42\\
DA & 3.66 & 3.71 & 3.73 & 3.90 & 3.75 & 3.89 & 3.94 & 3.63 & 3.82 &3.78\\
CtrlRegen$+$ & 3.17 & 3.20 & 3.32 & 3.41 & 3.27 & 3.40 & 3.55 & 3.26 & 3.43 &3.33\\
NFPA & 1.04 & 1.61 & 1.89 & 1.85 & 1.55 & 1.13 & 1.87 & 1.73 & 1.62 &1.58\\
Imprint & 3.00 & 3.09 & 3.08 & 3.35 & 3.20 & 3.16 & 3.34 & 2.85 & 3.22 &3.14\\
UnMarker & 3.09 & 2.84 & 2.95 & 3.05 & 3.13 & 3.13 & 2.92 & 3.00 & 3.04 &3.02\\
NPA & 3.18 & 3.25 & 3.24 & 3.72 & 3.82 & 3.74 & 3.75 & 3.33 & 3.66 &3.52\\

\midrule
\rowcolor{green!7}
\method & \textbf{3.66} & \textbf{3.72} & \textbf{3.75} & \textbf{4.09} & \textbf{3.93} & \textbf{4.24} & \textbf{4.11} & \textbf{3.80} & \textbf{4.06} & \textbf{3.93}\\
\rowcolor{green!7}
\methodA & \textbf{3.91} & \textbf{3.94} & \textbf{3.97} & \textbf{3.88} & \textbf{3.68} & \textbf{3.99} & \textbf{4.06} & \textbf{3.85} & \textbf{3.92} &\textbf{3.91}\\
\bottomrule
\end{tabular}
\end{table*}

\subsection{Evaluation Results}\label{sec:exp_res}
\subsubsection{Effectiveness of \method and \methodA} \label{sec:effectiveness}


We evaluate \method and \methodA in terms of watermark removal effectiveness, focusing on their ability to suppress detection scores across diverse watermarking schemes and against representative attack baselines. \Cref{tab:attack_vs_methods} summarizes watermark detection performance under different attacks. For multi-bit watermarking schemes, we report BA, while for the zero-bit TR, we report 
$\mathrm{TPR}@1\%\mathrm{FPR}$. In the absence of attacks, all watermarking methods achieve near-perfect detection or bit extraction accuracy, confirming the correctness of watermark embedding and detection.

Across eight multi-bit watermarking schemes, existing attacks remain largely ineffective: most baselines continue to achieve high BA (e.g., VA retains average BA of $80.17\%$, and UnMarker remains at $91.81\%$), indicating that the embedded watermarks are still reliably recoverable. VA and DA, which represent simple regeneration attacks, \textit{fail to compromise robust watermarking schemes} like VINE and SleeperMark. UnMarker, although effective against certain post-hoc watermarking, does not generalize beyond its targeted embedding artifacts. \textit{Notably, all the above approaches are ineffective against initial-noise-based watermarking schemes.} 

\method reduces the average BA to 53.14\%, approaching the random-guessing baseline (50\%), and achieves near-chance performance on robust schemes like GS and SleeperMark. Its amortized variant, \methodA, exhibits comparable performance. Overall, both attacks outperform existing baselines in disrupting multi-bit watermark detection, with \method achieving the strongest average effectiveness. For the zero-bit TR scheme, \method reduces $\mathrm{TPR}@1\%\mathrm{FPR}$ from $100\%$ to $43\%$, while \methodA achieves $63\%$. Although some baselines achieve slightly lower detection rates on TR, our methods provide broader and more consistent suppression across both multi-bit and presence-only detectors. These results highlight the strong generalizability of our approach. By explicitly decorrelating $z^w_T$ and $z^{\mathrm{(adv)}}_0$ through NLAS minimization, \method remains highly effective even against the stealthiest and most lossless initial-noise-based watermarking schemes (e.g., GS and PRC).

\begin{figure*}[t]
    \centering
    \captionsetup[subfigure]{skip=3pt, font=small}

    \begin{subfigure}[b]{0.19\linewidth}
        \centering
        \includegraphics[width=\linewidth]{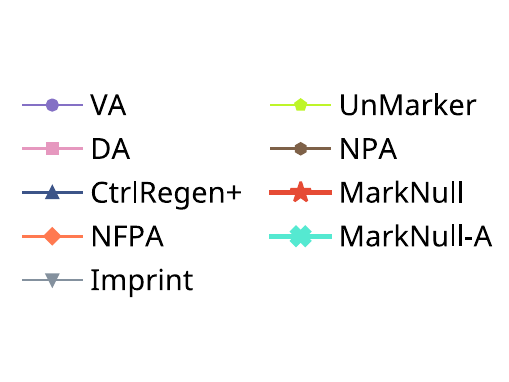}
        \caption{Legend}
        \label{fig:legend}
    \end{subfigure}\hfill
    \begin{subfigure}[b]{0.19\linewidth}
        \centering
        \includegraphics[width=\linewidth]{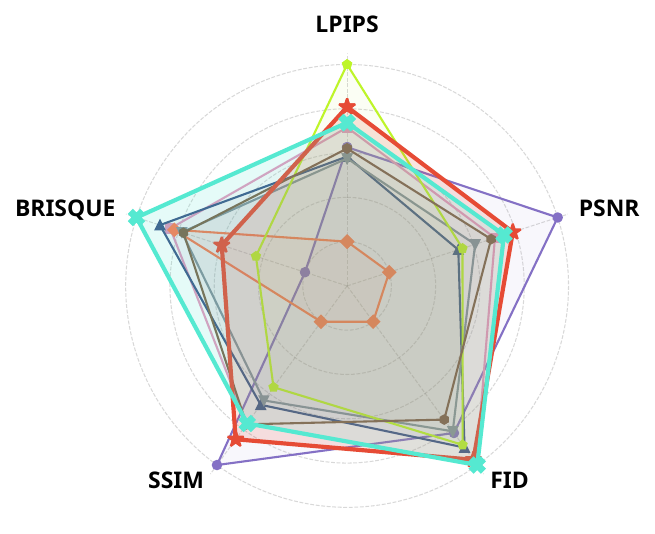}
        \caption{DwtDct}
        \label{fig:dwtdct_2}
    \end{subfigure}\hfill
    \begin{subfigure}[b]{0.19\linewidth}
        \centering
        \includegraphics[width=\linewidth]{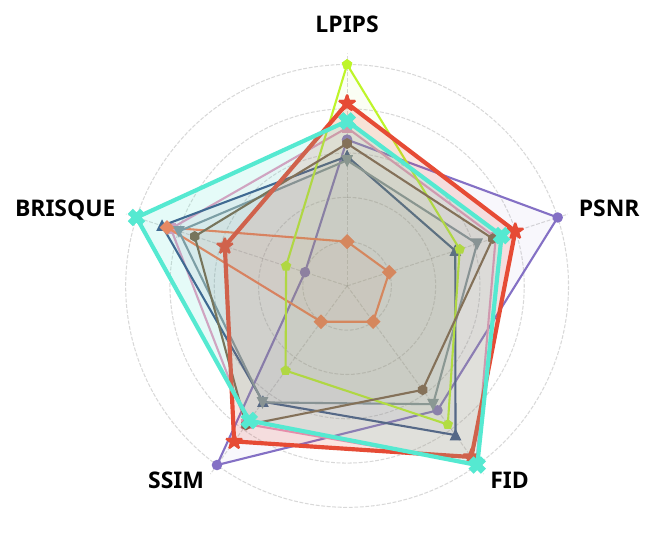}
        \caption{RivaGan}
        \label{fig:rivagan}
    \end{subfigure}\hfill
    \begin{subfigure}[b]{0.19\linewidth}
        \centering
        \includegraphics[width=\linewidth]{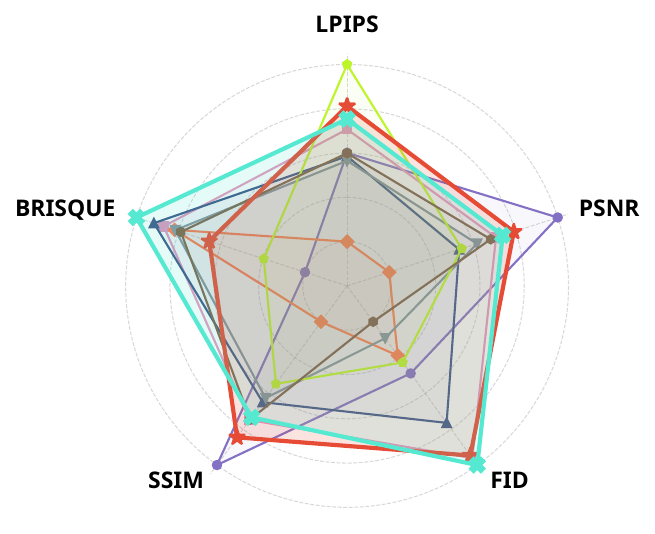}
        \caption{HiDDeN}
        \label{fig:hidden}
    \end{subfigure}\hfill
    \begin{subfigure}[b]{0.19\linewidth}
        \centering
        \includegraphics[width=\linewidth]{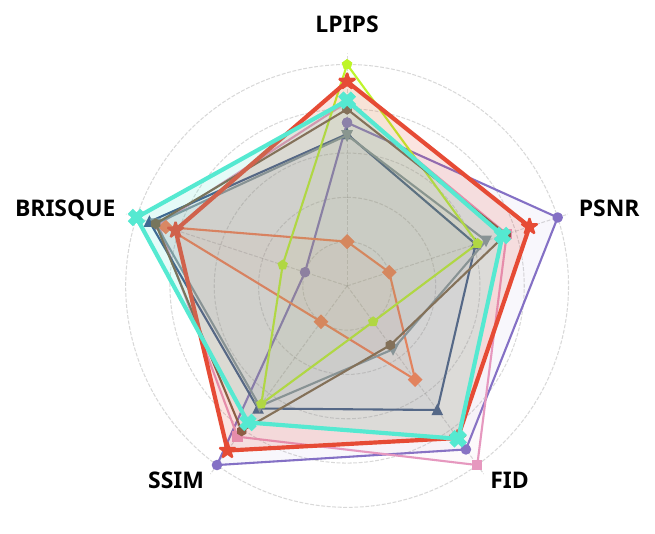}
        \caption{StableSignature}
        \label{fig:stablesignature}
    \end{subfigure}

    \begin{subfigure}[b]{0.19\linewidth}
        \centering
        \includegraphics[width=\linewidth]{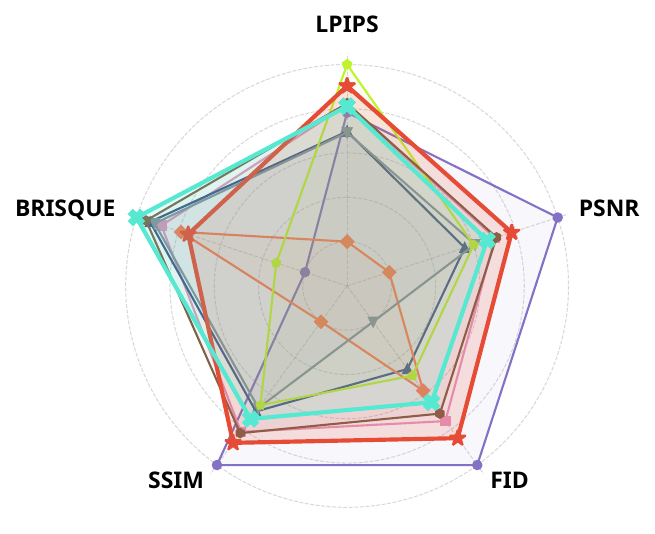}
        \caption{VINE}
        \label{fig:vine}
    \end{subfigure}\hfill
    \begin{subfigure}[b]{0.19\linewidth}
        \centering
        \includegraphics[width=\linewidth]{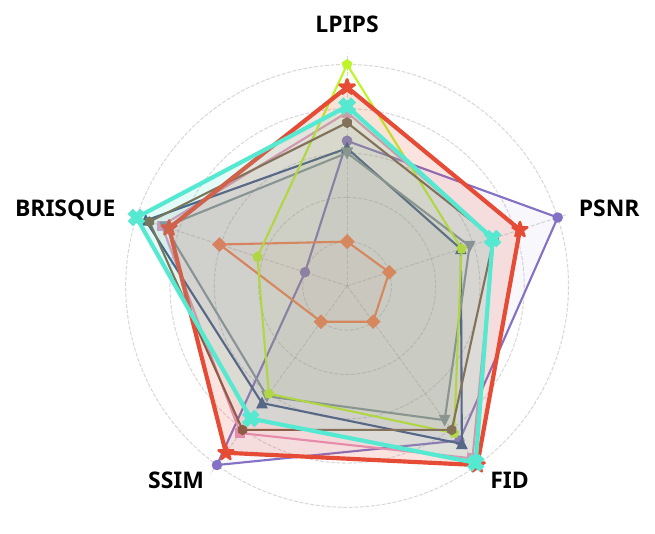}
        \caption{SleeperMark}
        \label{fig:SleeperMark}
    \end{subfigure}\hfill
    \begin{subfigure}[b]{0.19\linewidth}
        \centering
        \includegraphics[width=\linewidth]{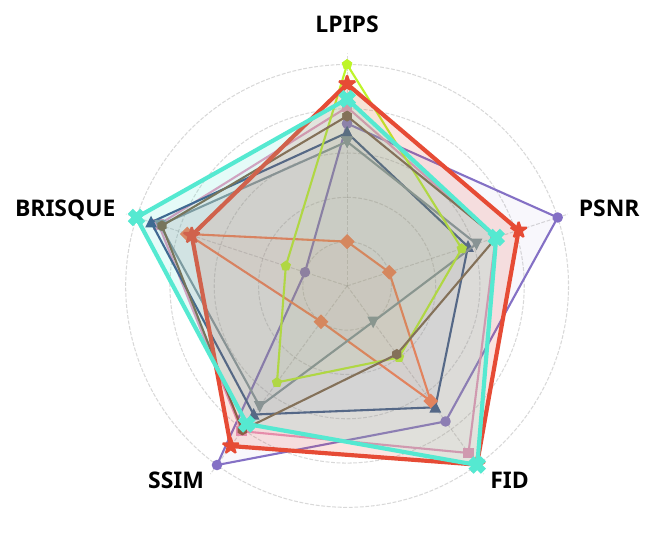}
        \caption{GS}
        \label{fig:gs}
    \end{subfigure}\hfill
    \begin{subfigure}[b]{0.19\linewidth}
        \centering
        \includegraphics[width=\linewidth]{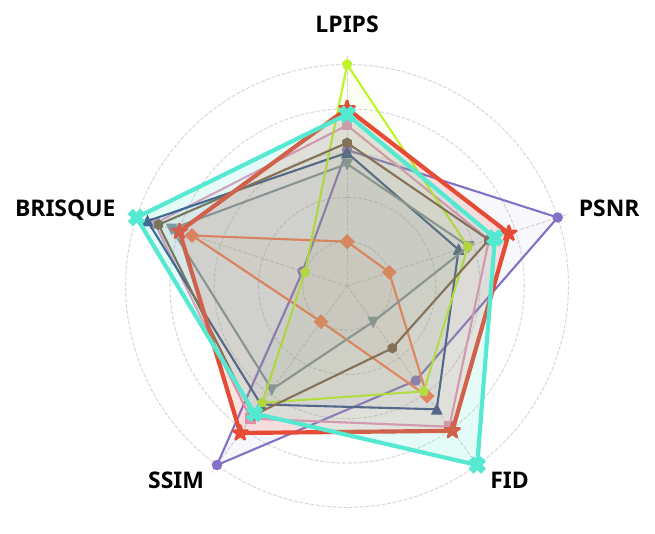}
        \caption{PRC}
        \label{fig:prc}
    \end{subfigure}\hfill
    \begin{subfigure}[b]{0.19\linewidth}
        \centering
        \includegraphics[width=\linewidth]{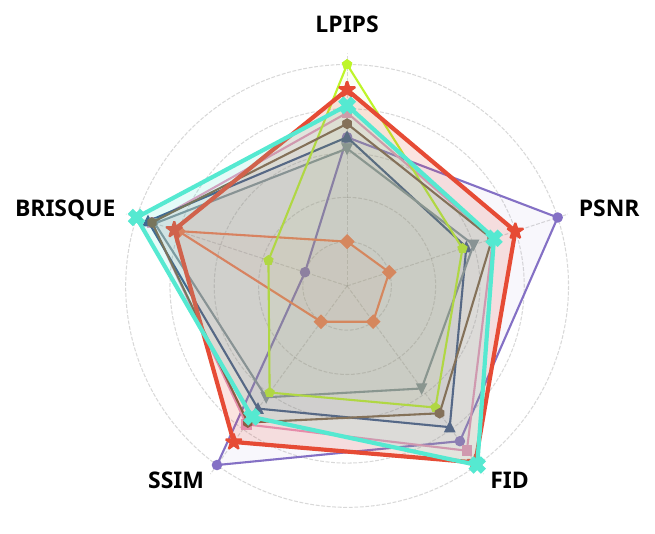}
        \caption{TR}
        \label{fig:tr}
    \end{subfigure}

    \caption{Quantitative comparison of image quality metrics. For visualization, we normalize all metrics and orient them so that larger values consistently indicate better quality (i.e., lower-is-better metrics are inverted).}
    \label{fig:all_radar_charts}
\end{figure*}

\subsubsection{Evaluation of Image Quality Preservation} 

We conduct a comprehensive quantitative evaluation against seven baseline attacks to demonstrate that \method and \methodA do not compromise the semantic and perceptual fidelity of images. \Cref{tab:CQS_vs_methods} shows that \method and \methodA demonstrate SOTA performance in maintaining image quality. \method achieves the highest CQS across all nine tested watermarking schemes, consistently outperforming existing methods such as UnMarker and NPA. In particular, for generative watermarking schemes like SS and TR, \method achieves a CQS of 4.09 and 4.06, respectively, significantly surpassing the strongest baseline DA. These results indicate that \method effectively suppresses watermark signals while preserving image fidelity, whereas baseline latent-space adversarial attacks (e.g., NPA and Imprint) often induce visible artifacts or blurring, leading to degraded perceptual quality. \Cref{fig:all_radar_charts} illustrates the performance distributions across individual quality metrics. The results show that \method and \methodA perform consistently well across both pixel-level metrics (PSNR, SSIM) and perceptual metrics (LPIPS, FID). Additional qualitative comparisons of attacked images are provided in \Cref{sec:quali_analy}.

\begin{takeawaybox}

\textbf{Takeaway 1:} Prior adversarial attacks suffer from an inherent fidelity–removal trade-off. \method effectively breaks this dilemma. By explicitly minimizing NLAS to decorrelate watermark signals from semantic content, it achieves near-perfect removal performance while maintaining SOTA perceptual quality. \end{takeawaybox}

\subsubsection{Attack Efficiency Assessment}
We evaluate the practical deployability of our attacks by measuring attack latency (seconds per image) and GPU memory consumption (VRAM in MB) on a single NVIDIA A100 GPU. As shown in \Cref{fig:efficiency}, existing optimization-based adversarial attacks (Imprint, UnMarker, and NPA) exhibit substantial computational costs. In particular, UnMarker requires over 700 seconds per image, and NPA takes 33.8 seconds, making both unsuitable for large-scale deployment. In contrast, our proposed \methodA achieves significantly improved efficiency, requiring only 0.50 seconds per image, which is orders of magnitude faster than all optimization-based baselines. In terms of memory usage, \methodA consumes 6,282 MB of VRAM, comparable to CtrlRegen$+$ (6,532 MB) and lower than high-resource methods such as UnMarker (14,378 MB). Although \method is more computationally intensive due to iterative optimization, its VRAM usage remains below 10 GB, suggesting its potential compatibility with consumer GPUs.

\begin{figure}[h]
    \centering
    \includegraphics[width=\linewidth]{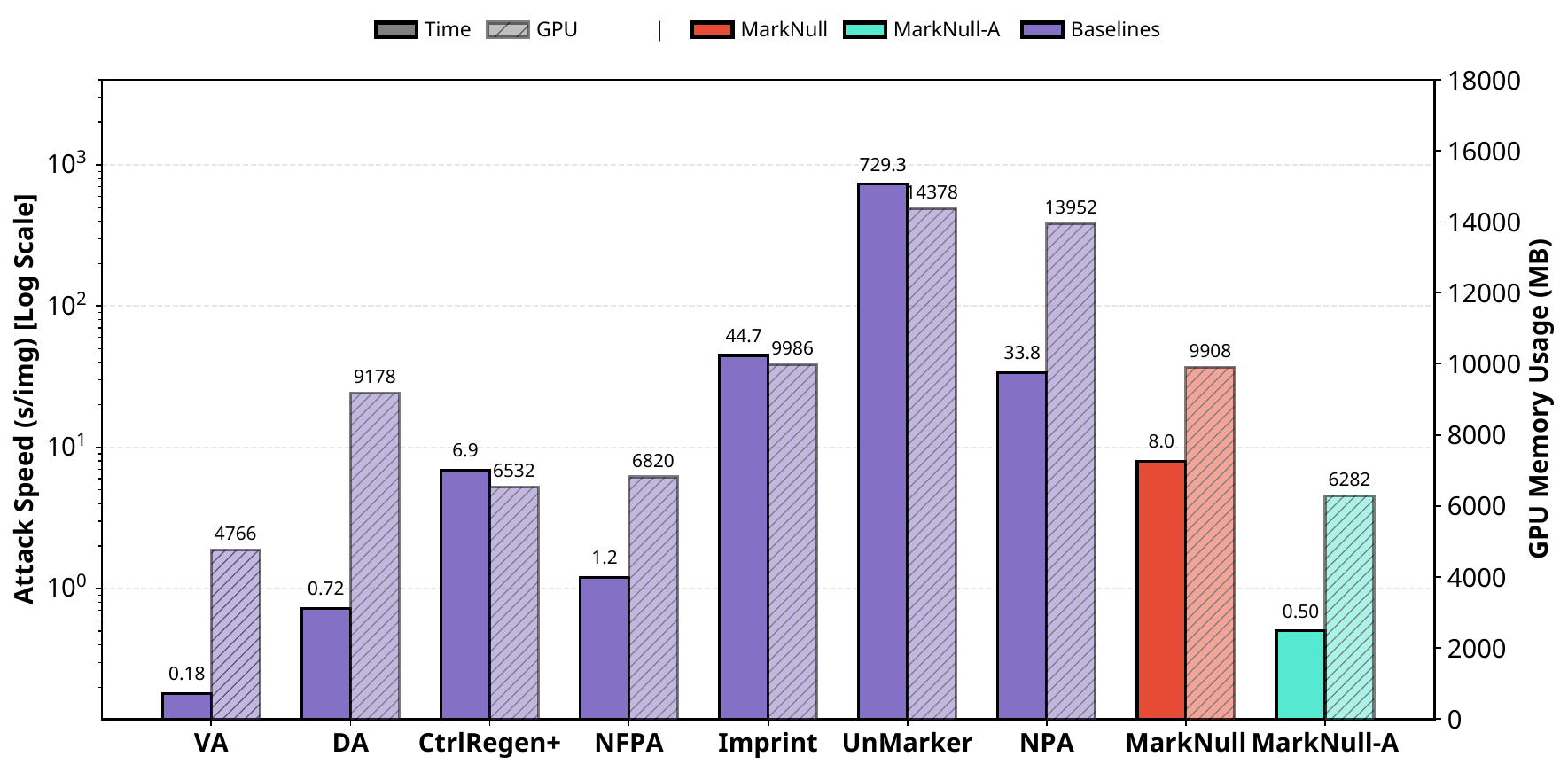}
    \caption{Comparison of attack speed and GPU memory usage across nine watermark removal methods.}
    \label{fig:efficiency}
\end{figure}

\begin{takeawaybox}

\textbf{Takeaway 2:} Existing attacks incur prohibitive computational overhead. \methodA addresses this by replacing iterative DDIM inversions with a single feed-forward pass, achieving real-time inference at low VRAM consumption. These results demonstrate that effective watermark removal requires no costly iterative refinement, enabling scalable deployment on consumer hardware.
\end{takeawaybox}

\subsection{Breaking Google's SynthID-Image} \label{sec:synthID}
SynthID-Image~\cite{gowal2025synthid} is a post-hoc, model-agnostic watermarking scheme designed for large-scale provenance verification of AI-generated images. It employs a learned encoder–detector pair to embed and detect imperceptible watermarks, and has been deployed across Google AI systems\footnote{\url{https://deepmind.google/models/synthid/}}. We evaluate our \method, \methodA, and baselines on the commercial SynthID-Image watermarking system using 20 AI-generated images produced by Imagen-3 with prompts from the SDP dataset \cite{gustavosta_stable_diffusion_prompts}. After applying each attack, we assess watermark detectability and image fidelity, with results reported in \Cref{tab:attack_synthid}. Since the official detection code of SynthID-Image is not publicly available, we follow Google’s Gemini (fast mode) verification procedure for watermark detection\footnote{\url{https://support.google.com/gemini/answer/16722517?hl=en-SG\&co=GENIE.Platform\%3DAndroid}}.

\Cref{tab:attack_synthid} reveals the critical security vulnerabilities of SynthID-Image. Our work is the \textit{first study} to explore the robustness of a commercial AI-generated image watermarking system against removal attacks. While SynthID-Image exhibits robustness against VA, it suffers a catastrophic failure when subjected to advanced attacks. We attribute this fragility to the limitations of its post-hoc embedding scheme, which proves insufficient against latent-space manipulations. However, among all methods achieving a 100\% ASR, \method reaches the highest CQS, as illustrated in \Cref{fig:synthid_img_qual}. Furthermore, the results for \method reported above are obtained using the default high-perturbation configuration described in \Cref{sec:marknull_set}. For post-hoc watermarks, a milder configuration (e.g., $\epsilon=1$, $N=10$, and $\kappa=0.01$) is sufficient to achieve a 100\% ASR while preserving high perceptual quality \textit{(LPIPS: 0.12, FID: 38.19, PSNR: 25.36, SSIM: 0.80, and BRISQUE: 6.11)}.


\begin{table}[t]
\centering
\small
\setlength{\tabcolsep}{4pt}
\renewcommand{\arraystretch}{0.9}
\caption{Attack evaluation against SynthID-Image. }
\label{tab:attack_synthid}
\begin{tabular}{p{2.5cm}*{2}{P{2.5cm}}}
\toprule

\textbf{Attack}  &ASR[\%] & CQS\\
\midrule
VA &0  &4.07 \\
DA  &100    &3.66\\
CtrlRegen$+$  &100 &3.29  \\
NFPA  &100 &1.21\\
Imprint &100 &3.21 \\
UnMarker  &100 &2.86 \\
NPA  &100 & 3.45  \\

\midrule
\rowcolor{green!7}
\method &\textbf{100} &\textbf{3.87}\\
\rowcolor{green!7}
\methodA &\textbf{100} &\textbf{3.70}\\
\bottomrule
\end{tabular}
\end{table}

\begin{figure}[t]
    \centering
    \includegraphics[width=0.69\linewidth]{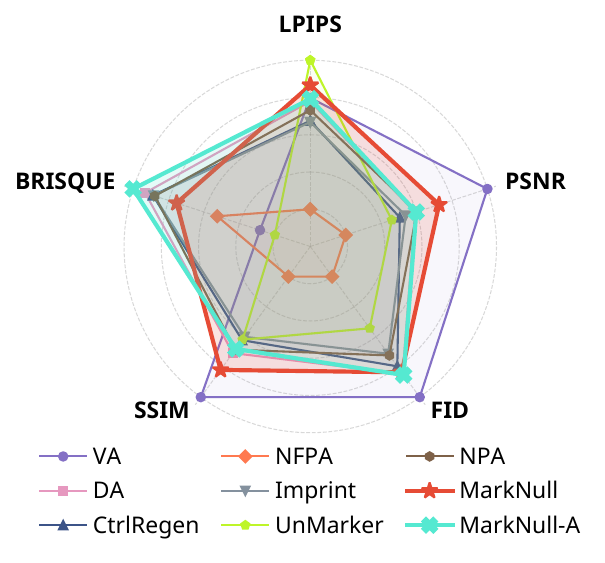}
    \caption{Image-quality metrics for SynthID-Image across different attack baselines and our methods.}
    \label{fig:synthid_img_qual}
\end{figure}


\subsection{Breaking AI-generated Video Watermarking} \label{sec:video}
To evaluate the scalability of our proposed attacks beyond static images, we extend our evaluation to emerging video watermarking frameworks: VideoShield~\cite{huvideoshield} and VideoMark~\cite{hu2025videomark}. VideoShield embeds watermarks by mapping bits into the initial noise of the diffusion process and uses DDIM inversion for watermark extraction. Similarly, VideoMark encodes watermarks into the initial noise with PRC codes on a per-frame basis and introduces a temporal matching module to enhance robustness against temporal distortions. 

Our evaluation is conducted using the \textit{damo-vilab/text-to-video-ms-1.7b}\footnote{\url{https://huggingface.co/ali-vilab/text-to-video-ms-1.7b}} text-to-video model. We randomly select 20 prompts from the test set of VBench~\cite{huang2024vbench}, and generate 16-frame videos at a resolution of $256 \times 256$ using 25 inference steps. For baseline video watermarking attacks, we adapt existing image-based watermark removal methods by applying them independently to each video frame, and then evaluate the resulting video quality based on the processed frames. Following~\cite{pan2025markdiffusion}, we use MUSIQ~\cite{ke2021musiq} as a no-reference perceptual quality metric, averaging frame-level scores to assess overall video quality.

As detailed in \Cref{tab:attack_video}, video watermarking schemes exhibit strong robustness: baseline attacks such as VA and DA fail to disrupt initial-noise-based watermarking, retaining BA above 90\%. In contrast, \method exhibits superior transferability, successfully suppressing BA to approximately 50\%. \methodA achieves an optimal trade-off, attaining the highest video quality (MUSIQ $> 0.6$) among successful attacks while effectively compromising watermark verification. Although NFPA and NPA reduce VideoMark’s BA to very low values (24.26\% and 17.26\%), this does not indicate superior watermark removal. For multi-bit watermarking, BA values significantly below the 50\% random-guessing baseline reflect bit inversion rather than erasure: the watermark signal remains present but becomes negatively correlated, allowing defenders to recover the original message by simply flipping the decoded bits (e.g., $\sim83\%$ effective recovery for NPA). 


\begin{table}[htbp]
\centering
\small
\setlength{\tabcolsep}{4pt}
\renewcommand{\arraystretch}{0.9}
\caption{Attack evaluation against video watermarking. }
\label{tab:attack_video}
\begin{tabular}{p{2.5cm}*{4}{P{1cm}}}
\toprule

\multirow{2}{*}{\textbf{Attack}} & \multicolumn{2}{c}{\textbf{VideoShield}} &\multicolumn{2}{c}{\textbf{VideoMark}} \\
\cmidrule(lr){2-3} \cmidrule(lr){4-5}
 &BA[\%]&MUSIQ &BA[\%]&MUSIQ\\
 \midrule
\textit{No Attack}  &100 &0.66  &100&0.62\\

\midrule

VA      &99.95 &0.59 &93.81 &0.55   \\
DA      &97.31 &0.61 &66.06 &0.59 \\
CtrlRegen$+$ &91.20 &0.59 &56.80 &0.57  \\
NFPA    &56.93 &0.65 &24.26&0.59 \\
Imprint  &57.74 &0.64&52.79&0.61\\
UnMarker  &80.79&0.42 &80.74&0.48 \\
NPA  &82.71&0.47&17.26&0.44  \\

\midrule
\rowcolor{green!7}
\method &\textbf{47.12}&0.57 &\textbf{50.30}&0.55\\
\rowcolor{green!7}
\methodA &76.46 &\textbf{0.65} &\textbf{49.48}&\textbf{0.61}\\
\bottomrule
\end{tabular}
\end{table}

\begin{takeawaybox} 

\textbf{Takeaway 3:} \method and \methodA \textit{demonstrate unprecedented scalability across the broadest range of watermarking paradigms}, achieving a 100\% ASR against Google's SynthID-Image and generalizing seamlessly to video watermarking (VideoShield, VideoMark) without modality-specific adaptation. This confirms their truly model-agnostic nature, establishing a novel benchmark for universal watermark robustness evaluation.
\end{takeawaybox}

\section{Watermark Removal Defense} \label{sec:defense}
In this section, we explore a potential defense strategy against advanced watermark removal attacks by leveraging the distributional properties of LDMs. These models are trained to capture the high-probability distribution of natural images, effectively defining a low-dimensional generative manifold shaped by the model prior.

Watermark removal attacks typically add a subtle yet unnatural pixel-space perturbation $\Delta_A = I_A - I_w$ to suppress the watermark signal. Although such perturbations may be visually imperceptible, they often violate the learned prior and push the sample off the generative manifold in a distributional sense. Motivated by the theoretical intuition of DiffPure \cite{nie2022diffusion}, we perform \textit{Inversion $\rightarrow$ Regeneration} for attack detection. As shown in \Cref{alg:detection}, given a query image $I_Q$, we perform a full DDIM inversion to obtain $\hat{z}^Q_{T}$ and reconstruct the image as $I_Q^{\mathrm{rec}}$. For clean images that lie close to the natural image manifold, the inversion and regeneration process is approximately self-consistent, resulting in a small discrepancy between $I_Q^{\mathrm{rec}}$ and $I_Q$. In contrast, when $I_Q$ contains off-manifold adversarial perturbations in the pixel space, the diffusion prior tends to suppress these unnatural components during regeneration, leading to a noticeably larger reconstruction gap. To quantify this, we adopt the LPIPS reconstruction error $\mathcal{L}_\mathrm{LPIPS}(I_Q^{\mathrm{rec}}, I_Q)$ as the detection statistic. A query image is flagged as attacked if the LPIPS exceeds a predefined threshold $\tau_\mathrm{dete}$. This threshold is selected based on analyzing the empirical distribution of regeneration errors on a held-out validation set.

\begin{algorithm}[t]
    \caption{Watermark Removal Attack Detection}
    \label{alg:detection}
    
    \SetKwInOut{Input}{Input}
    \SetKwInOut{Output}{Output}
    
    \Input{Query image $I_Q$, Target Model $\Theta$, Detection Threshold $\tau_\mathrm{dete}$, LPIPS Distance $\mathcal{L}_\mathrm{LPIPS}$.}
    \Output{Binary decision: \textit{Attacked} or \textit{Clean}}


    \textcolor{blue}{\tcp{1. Regeneration}}
    $z^Q_0=\mathcal{E}(I_Q)$\;
    $\hat{z}^Q_{T} =\mathsf{DDIM}_{0\rightarrow T}(z^Q_0;\, \mathcal{U})$\;
    $z_0^{Q,\mathrm{rec}} =\mathsf{DDIM}_{T\rightarrow 0}(\hat{z}^Q_{T};\, \mathcal{U})$\;
    $I_Q^{\mathrm{rec}}=\mathcal{D}(z_0^{Q,\mathrm{rec}})$\;
    
    \textcolor{blue}{\tcp{2. Error Calculation}}

    $S_{error} \leftarrow \mathcal{L}_\mathrm{LPIPS}(I_Q, I_Q^{\mathrm{rec}})$\;
    \textcolor{blue}{\tcp{3. Attack Detection}}

    \If{$S_{error} > \tau_\mathrm{dete}$}{
    \Return{\textit{Attacked}}
    }
    \Else{
        \Return{\textit{Clean}}\; 
    }
\end{algorithm}

To evaluate our defense, we formulate attack detection as a binary classification task between clean watermarked images and attacked watermark-removed images. We use the reconstruction error, measured by the LPIPS distance $S_{\mathrm{error}}$, as the detection statistic. The decision threshold $\tau_{\mathrm{dete}}$ is selected via the Youden Index by sweeping observed scores to construct the Receiver Operating Characteristic (ROC) curve:

\begin{equation}
\begin{aligned}
\mathrm{FPR}(\tau_{\mathrm{dete}}) &= P\left(S_{\mathrm{error}} > \tau_{\mathrm{dete}} \mid \mathrm{Clean}\right),\\
\mathrm{TPR}(\tau_{\mathrm{dete}}) &= P\left(S_{\mathrm{error}} > \tau_{\mathrm{dete}} \mid \mathrm{Attacked}\right).
\end{aligned}
\end{equation}

The optimal threshold is then given by:
\begin{equation}
\tau_\mathrm{dete}^* = \operatorname*{argmax}_{\tau_\mathrm{dete}} (\text{TPR}(\tau_\mathrm{dete}) - \text{FPR}(\tau_\mathrm{dete})),
\end{equation}
yielding a balance between detection sensitivity and false alarms without using arbitrary error margins.

As shown in \Cref{fig:def_marknull,fig:def_marknullA}, attacked images generally exhibit higher LPIPS reconstruction errors than clean watermarked images under both \method and \methodA. This separation supports our hypothesis that adversarial perturbations introduce deviations that are amplified during DDIM inversion and regeneration, making reconstruction error a useful forensic signal. The resulting detector achieves promising performance, particularly for in-generation watermarking schemes. For example, it achieves a TPR above 90\% for SleeperMark while maintaining a low FPR, indicating clear separation between clean and attacked samples.

Nevertheless, partial distribution overlap persists under stealthy attacks such as \methodA, making consistent detection challenging. This validates the imperceptibility of our attacks and highlights the need for more robust forensic defenses in future work.



\begin{figure}[h]
    \centering
    \includegraphics[width=\linewidth]{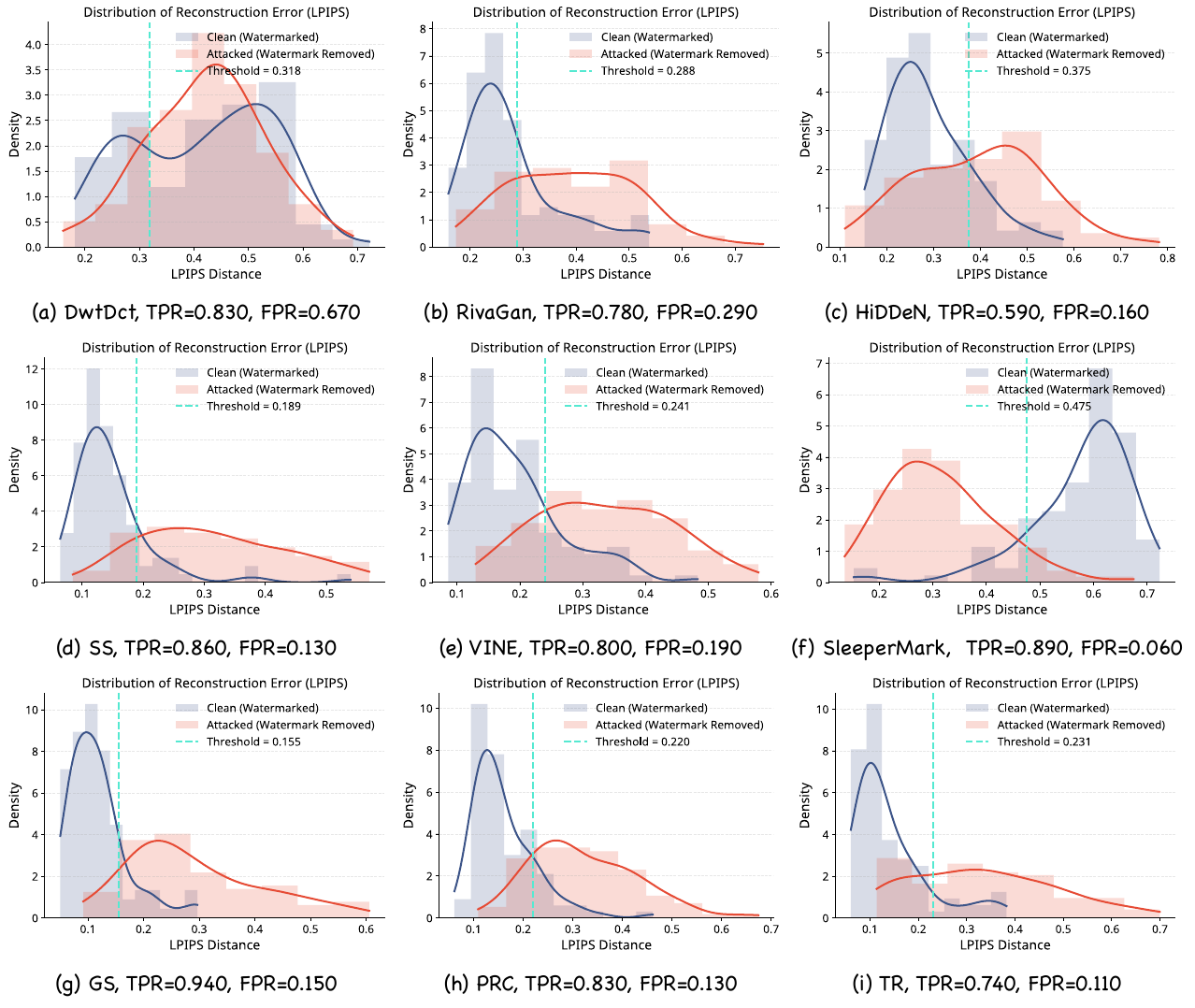}
    \caption{Comparison of LPIPS reconstruction error distributions and detection thresholds under \method (watermarked vs. watermark-removed).}
    \label{fig:def_marknull}
\end{figure}

\begin{figure}[h]
    \centering
    \includegraphics[width=\linewidth]{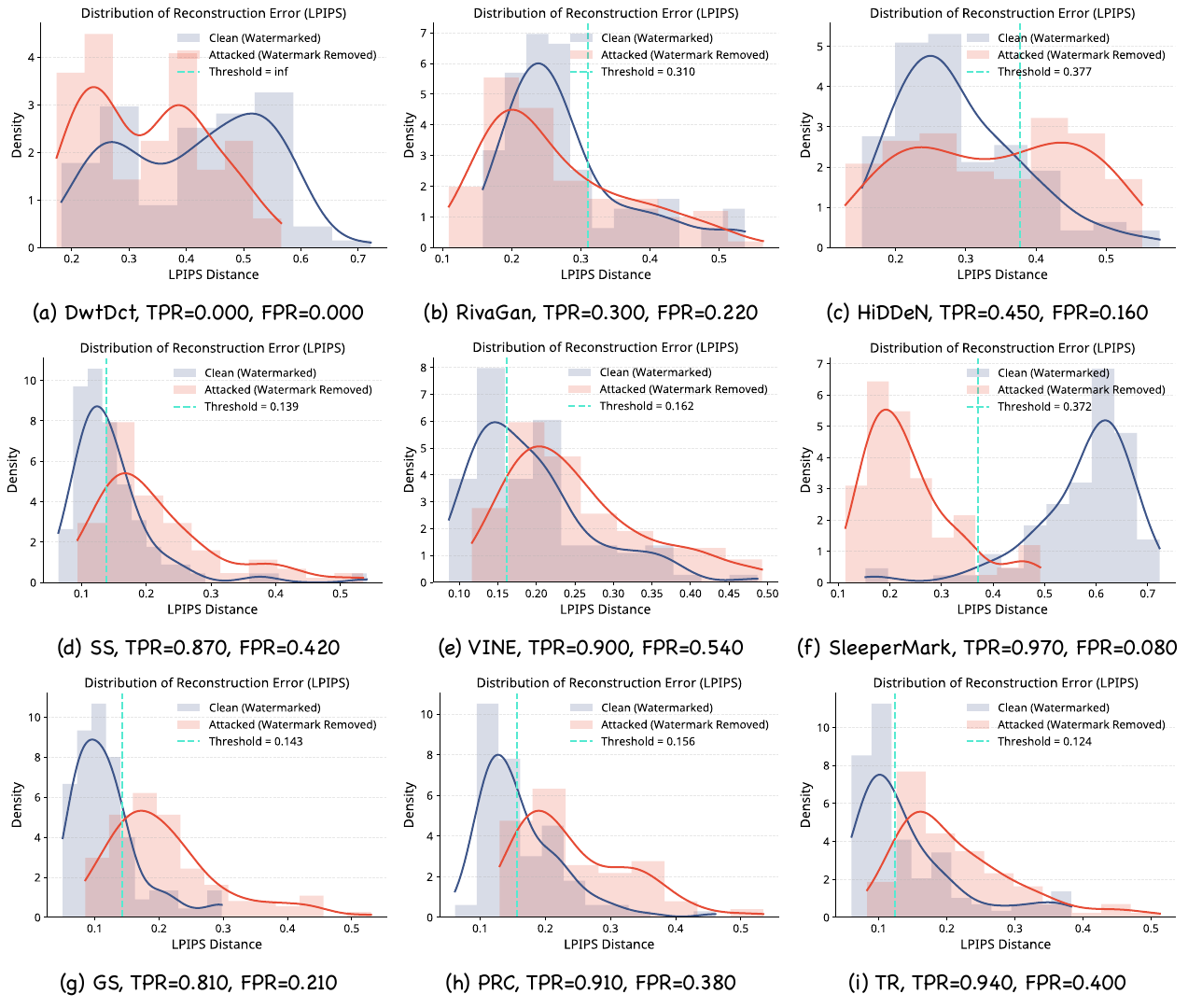}
    \caption{Comparison of LPIPS reconstruction error distributions and detection thresholds under \methodA (watermarked vs. watermark removed).}
    \label{fig:def_marknullA}
\end{figure}



\section{Conclusion}

In this paper, we have proposed \method and its amortized variant, \methodA, which together form a model-agnostic watermark removal attack framework that fundamentally overcomes the fidelity–removal trade-off and computational bottlenecks inherent in prior work. Specifically, \method leverages NLAS to guide an optimization process that provably decorrelates watermark signals from the initial noise via on-manifold latent manipulation, reducing the BA to 53.14\%, approaching the 50\% random-guessing baseline in watermark verification, while preserving semantic and perceptual fidelity. To support scalable deployment, \methodA distills the optimization objective into a learnable WRN, achieving fast inference speeds of 0.50 s/image while maintaining strong removal effectiveness. Extensive evaluations demonstrate that both \method and \methodA effectively remove watermarks across all major watermarking paradigms, and successfully compromise Google’s SynthID-Image (100\% ASR), as well as emerging video watermarking schemes. This attack framework, together with the proposed defense, serves as a rigorous red-teaming benchmark, encouraging the community to move beyond simplistic initial-noise-based watermarking and toward next-generation provenance tracking mechanisms that are robust to latent-space manipulation.

\section*{Acknowledgments}
We are grateful to our shepherd and anonymous reviewers for their valuable guidance and insightful comments. This research was supported by the Natural Sciences and Engineering Research Council of Canada Discovery Grant, the Canada Research Chair Program, National Cybersecurity Consortium R\&D Program, and the NVIDIA Academic Grant Program using A100 GPU-hours.

\section*{Ethical Considerations}
This paper investigates the robustness limitations of AI-generated image watermarking by proposing \method and \methodA, which remove watermarks through on-manifold latent manipulation under a no-box, model-agnostic threat model. Our goal is to advance rigorous watermark security evaluation and drive the development of more resilient provenance mechanisms, not to facilitate misuse.

\noindent \textbf{Potential misuse.} Watermark removal techniques can be misused to evade provenance tracking, undermine copyright attribution, and facilitate AI-generated misinformation. The risk is amplified by our attacks' model-agnostic nature, requiring no watermark information or detector access.

\noindent \textbf{Risk mitigation.} We take several steps to reduce misuse potential: (i) experiments are conducted on public datasets without private user data; (ii) we describe methodology at the level necessary for reproducibility, avoiding deployment-oriented instructions that would lower the barrier to attacking production systems; and (iii) we include a defensive counterpart, an attack detection mechanism to encourage balanced red-team and blue-team development.

\noindent \textbf{Vulnerability Disclosure.} In line with standard security research ethics, we will disclose our findings to Google's security team regarding SynthID-Image robustness issues. Our evaluation was conducted strictly within the official SynthID-Image API bounds, deliberately rate-limited, and explicitly designed to avoid any denial-of-service behavior.

\noindent \textbf{Broader impacts.} Our work provides a stronger benchmark for evaluating watermark robustness and highlights that robustness should not rely solely on pixel-domain representations when latent-space manipulation is feasible. We acknowledge disclosure risks, but believe transparent evaluation of failure modes, accompanied by effective defenses, is essential for the responsible deployment of generative models.

\section*{Open Science}
All source code, trained WRN weights, and test samples are publicly available on Zenodo: \url{https://doi.org/10.5281/zenodo.20201878}, accompanied by a detailed step-by-step \texttt{README.md}. Since all models and watermarking schemes evaluated in this paper are publicly released, our artifacts provide a complete basis for reproducing the main experimental results. We fully support open science and release these artifacts to facilitate reproducibility and foster further research.

\bibliographystyle{plain}
\bibliography{mybib}

\appendix

\section{Noise-Latent Alignment Score} \label{sec:nlas}

To evaluate the correlation within the latent space, we computed the NLAS for the optimized pairs $(z_0, z_T)$ against the randomly initialized $(z_0, z_{\mathrm{rand}})$, where $z_{\text{rand}} \sim \mathcal{N}(\mathbf{0}, \mathbf{I})$. As shown in \Cref{fig:nlas}, the NLAS between $z_0$ and its corresponding initial noise $z_T$ is markedly higher than that between $z_0$ and a randomly initialized latent $z_{\mathrm{rand}}$. Specifically, the $(z_0, z_T)$ pairs exhibit consistently elevated NLAS values with a stable median and interquartile range, whereas the NLAS of $(z_0, z_{\mathrm{rand}})$ concentrates near zero with negligible variance across samples. These observations indicate that NLAS effectively captures the structured dependence between corresponding latents along the same generation/inversion trajectory, while such dependence is absent when pairing $z_0$ with independent Gaussian noise. This distinct separation underscores the utility of NLAS as a discriminative metric for identifying true latent correspondences.

\begin{figure}[h]
    \centering
    \includegraphics[width=0.7\linewidth]{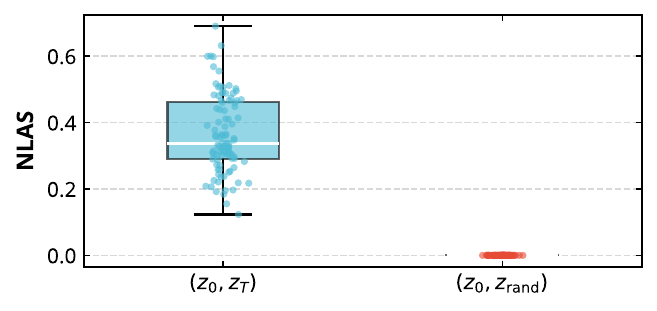}
    \caption{NLAS distributions for matched $(z_0, z_T)$ and random $(z_0, z_{\mathrm{rand}})$ pairs.}
    \label{fig:nlas}
\end{figure}

\section{Proposed Attack Details}   \label{sec:marknull2_set}
\subsection{\method}\label{sec:marknull_set}

We set the loss weights in \cref{eq:loss1} to $\lambda_1 = 1$, $\lambda_2 = 20$, $\lambda_3 = 1.5$, and $\lambda_4 = 3$. The adversarial latent variable $z_0^{(\mathrm{adv})}$ is optimized using the Adam optimizer with a learning rate $\kappa = 0.05$. To ensure the imperceptibility of adversarial perturbations, we impose a distortion budget of $\varepsilon = 1$. The optimization is performed over $N = 50$ gradient descent steps, obtaining a balance between attack effectiveness and computational efficiency.

To investigate the impact of different learning rates and optimization steps on \method's attack effectiveness and image quality, we conduct two controlled experiments. In the first setting, we fix the learning rate at $\kappa = 0.05$ and vary the number of gradient steps $N \in \{5, 10, 25, 50, 100\}$. In the second setting, we fix the number of steps at $N = 50$ and vary the learning rate $\kappa \in \{0.001, 0.01, 0.05, 0.1, 0.2, 0.5\}$. These experiments facilitate a systematic analysis of the trade-off between attack strength and perceptual fidelity under different optimization parameter settings.

Experimental results are shown in \Cref{fig:N_LPIPS_BA} and \Cref{fig:lr_LPIPS_BA}. These results reveal a clear trade-off between attack strength and perceptual fidelity. When varying the number of gradient steps $N$ (see \Cref{fig:N_LPIPS_BA}), increasing $N$ from 5 to 25 leads to notable gains in attack success (lower BA and TPR@1\%FPR), accompanied by increased perceptual deviation (higher LPIPS). However, beyond $N \geq 50$, improvements become marginal: LPIPS plateaus or even slightly decreases, suggesting that once the optimization nears convergence, additional iterations do not further degrade image quality and refine perturbations along more natural directions under the model prior. Correspondingly, the attack strength converges and becomes relatively stable for $N$ in the range of approximately $25$ to $50$, indicating diminishing returns from further optimization steps.

\begin{figure}[h]
    \centering
    \captionsetup[subfigure]{skip=3pt, font=small} 

    \begin{subfigure}[b]{0.45\linewidth}
        \centering
        \includegraphics[width=\linewidth]{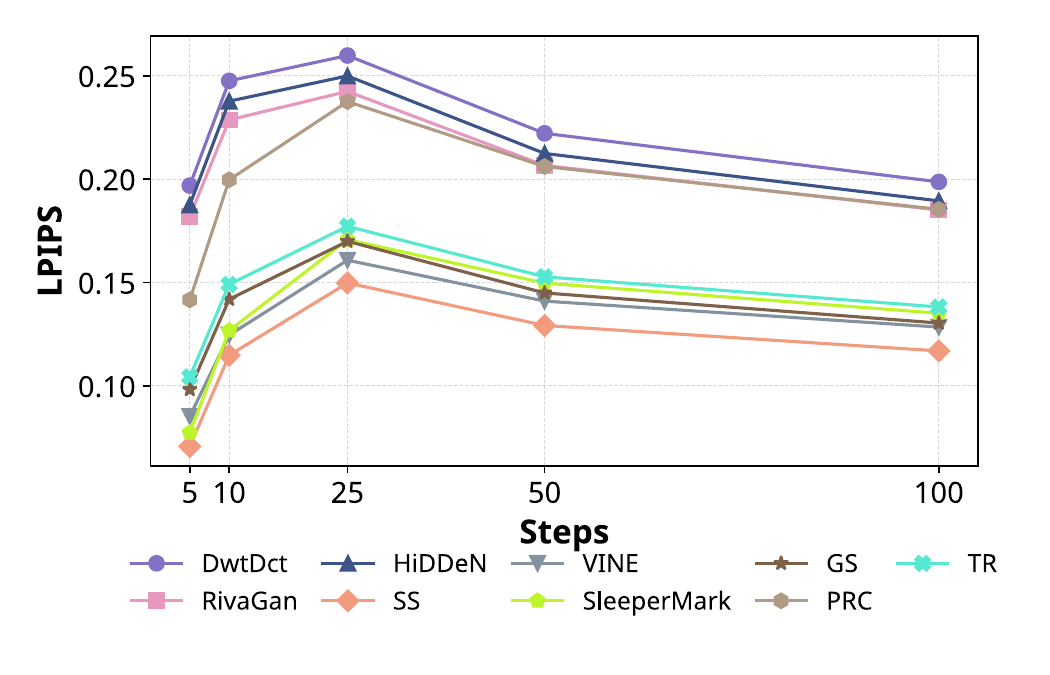}
        \caption{LPIPS}
        \label{fig:LPIPS_steps}
    \end{subfigure}
    \hfill
    \begin{subfigure}[b]{0.45\linewidth}
        \centering
        \includegraphics[width=\linewidth]{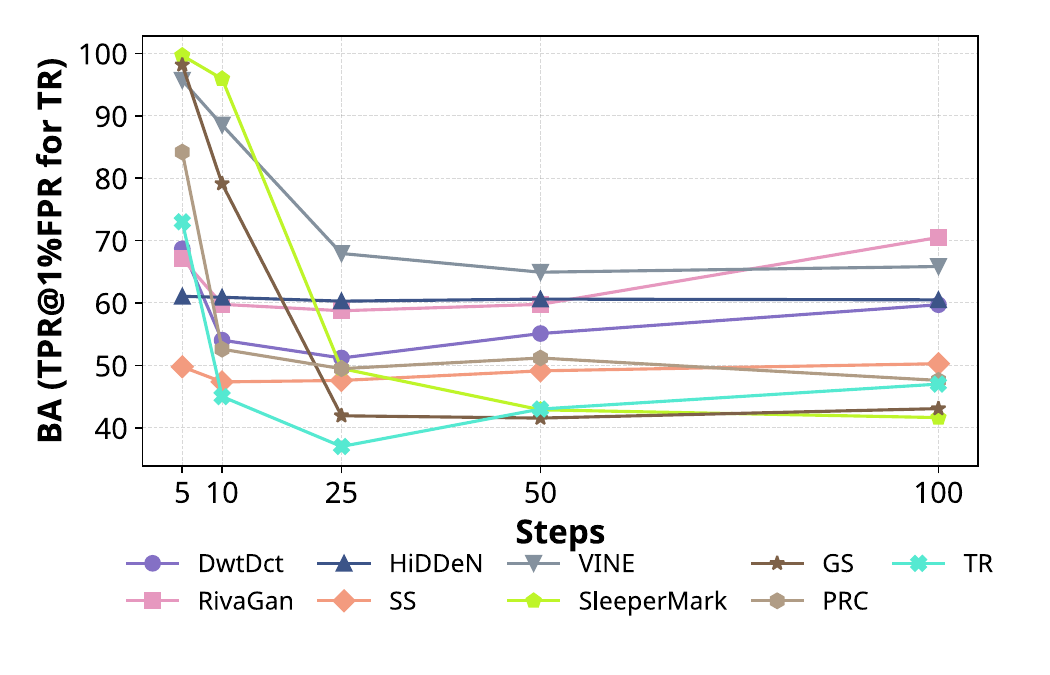}
        \caption{BA or TPR@1\%FPR}
        \label{fig:BA_steps}
    \end{subfigure}
    \hfill
    \caption{Effect of optimization steps ($N$) on attack strength and perceptual fidelity.}
    \label{fig:N_LPIPS_BA}
\end{figure}

\begin{figure}[h]
    \centering
    \captionsetup[subfigure]{skip=3pt, font=small} 

    \begin{subfigure}[b]{0.45\linewidth}
        \centering
        \includegraphics[width=\linewidth]{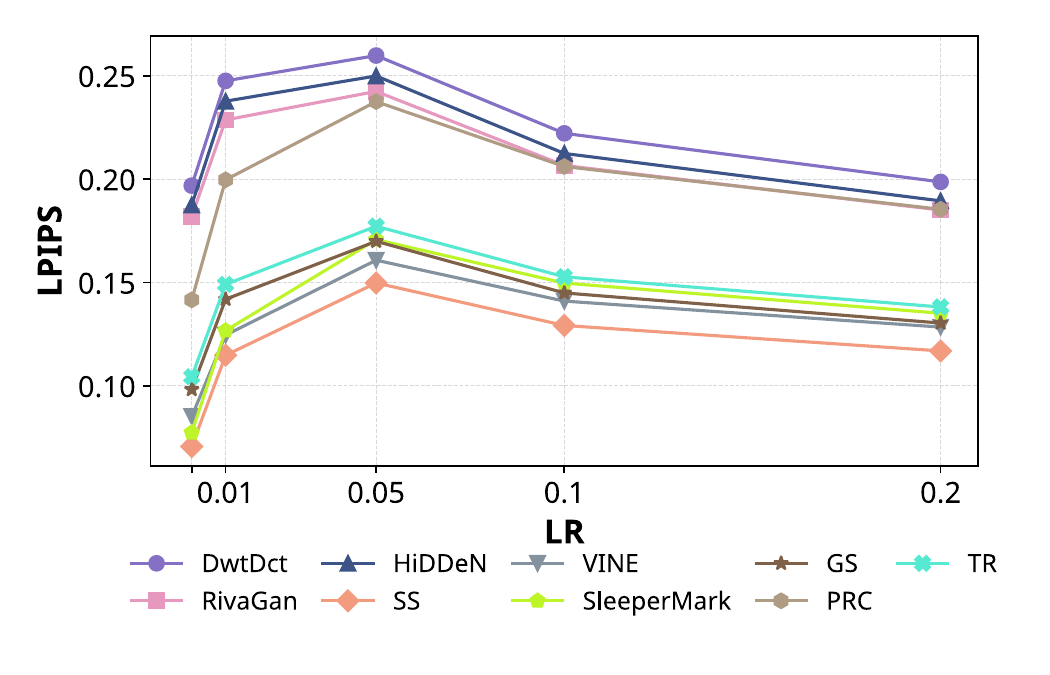}
        \caption{LPIPS}
        \label{fig:LPIPS_lr}
    \end{subfigure}
    \hfill
    \begin{subfigure}[b]{0.45\linewidth}
        \centering
        \includegraphics[width=\linewidth]{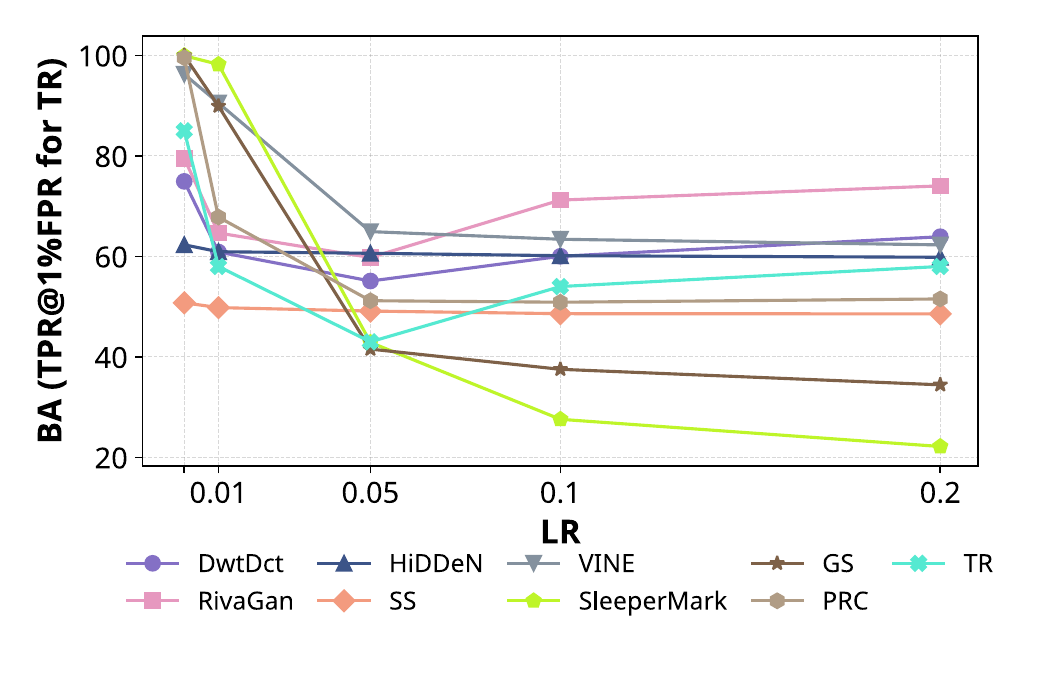}
        \caption{BA or TPR@1\%FPR}
        \label{fig:BA_lr}
    \end{subfigure}
    \hfill
    \caption{Effect of learning rate ($\kappa$) on attack strength and perceptual fidelity.}
    \label{fig:lr_LPIPS_BA}
\end{figure}


A similar trend emerges when adjusting the learning rate $\kappa$ with a fixed step count (see \Cref{fig:lr_LPIPS_BA}). Moderate learning rates often induce the most perceptual change, while larger $\kappa$ values lead to more stable solutions with comparable or improved attack performance. Sensitivity to $\kappa$ varies across watermarking schemes: some methods achieve significant BA reductions at moderate $\kappa$, while others require higher $\kappa$ to attain similar effectiveness. It is important to note that BA is not strictly \textit{``the lower, the better."} An excessively low BA often indicates that the attack drives the watermark detector’s extraction toward the bitwise complement $\overline{w}$, rather than erasing the watermark signal. In this case, the detector may still confidently report a watermark (e.g., SleeperMark in \cref{fig:BA_lr}). Therefore, for bit-string watermark removal, the most desirable outcome is a BA close to 50\%, implying that the extracted bit string is indistinguishable from random guessing and the watermark detectability is effectively neutralized.

Overall, these findings suggest that a moderate number of optimization steps ($N=50$) combined with a properly tuned learning rate ($\kappa=0.05$), as used in \Cref{sec:exp_res}, provides a favorable trade-off, achieving strong watermark removal with minimal perceptual degradation.

\subsection{\methodA} 
The WRN in \methodA adopts a hierarchical U-Net architecture based on Restormer~\cite{zamir2022restormer}\footnote{\url{https://github.com/leftthomas/Restormer}.}, featuring a four-level encoder-decoder with symmetric skip connections. Each block integrates MDTA for global context modeling and GDFN for local feature refinement. The encoder and decoder employ pixel-unshuffle and pixel-shuffle for downsampling and upsampling, respectively, with $1\times1$ convolutions for feature fusion. The network concludes with a refinement stage followed by a convolutional layer that predicts the residual image for high-fidelity restoration. We set $\lambda_5 = 1$, $\lambda_6 = 20$, $\lambda_7 = 1$ and $\lambda_8 = 4$ in \cref{eq:loss2}. The WRN is optimized for 100 epochs with a batch size of 7 using the Adam optimizer. The initial learning rate is set to $1 \times 10^{-3}$. A learning rate scheduler is applied, consisting of a linear warm-up during the first 10\% of the training epochs, followed by a cosine decay that gradually reduces the learning rate to $1 \times 10^{-6}$. The training set for the WRN consists of 2,000 images at a resolution of $512\times512$, generated by the proxy model SD1.5 using 2,000 randomly sampled prompts from the SDP dataset. We train the WRN on randomly cropped image patches of size $256 \times 256$.

\section{Attack Effectiveness Metrics}\label{sec:metric_detail}
\subsection{Watermark Detection}
The watermark length of each evaluated scheme is summarized in \Cref{tab:cap_thresh}. Let $w \in \{0,1\}^k$ be the embedded $k$-bit watermark and $\hat{w} \in \{0,1\}^k$ the bit string extracted from a query image. Detection is performed by comparing $\hat{w}$ with $w$ using the matching count $M$, defined as the complement of the Hamming distance:

\begin{equation}
M(w, \hat{w}) = \sum_{i=1}^{k} \mathbb{I}\left[w_i = \hat{w}_i\right],
\end{equation}
where $\mathbb{I}[\cdot]$ is the indicator function. A watermark is detected if the number of matches exceeds a predetermined threshold $\tau$, i.e., $M(w, \hat{w}) > \tau$.

Formally, the detection is formulated as a hypothesis testing. Under the null hypothesis $\mathcal{H}_0$, the image does not contain the target watermark $w$, and the extracted bits $\hat{w}=\{\hat{w}_i\}_{i=1}^k$ are modeled as i.i.d. Bernoulli variables. The False Positive Rate (FPR) can be defined as:
\begin{equation}
\mathbb{P}\left(M(w,\hat{w})>\tau \mid \mathcal{H}_{0}\right)= \frac{1}{2^{k}} \sum_{i=\tau+1}^{k} \binom{k}{i}
= I_{1/2}\left(\tau+1,\,k-\tau\right),
\end{equation}
which follows the tail probability of a binomial distribution and admits a closed-form expression via the regularized incomplete beta function $I_p(a,b)$.

The i.i.d.\ assumption and binomial detection model have been empirically validated in prior works~\cite{fernandez2023stable,lu2024robust}. Bit Accuracy (BA) is commonly used to measure watermark detectability and is defined as $\mathrm{BA} = M(w, \hat{w})/k$.

\begin{table}[h]
\centering
\caption{Watermark length and detection threshold of the evaluated methods.}
\label{tab:cap_thresh}
\renewcommand{\arraystretch}{0.92}
\begin{tabular}{lcccc}
\toprule
\textbf{Method}  & \textbf{Len.} (bits) & \textbf{BA Thres. $\tau/k$} [\%] & \textbf{FPR} \\
\midrule
DwtDct         & 48 &69 &0.01 \\
RivaGan                &32              &75 & 0.01 \\
HiDDeN                & 48              &69 & 0.01\\
SS                     & 48 &69 &0.01 \\
VINE                          & 96             & 62 & 0.01 \\
SleeperMark                      & 48              &69 & 0.01 \\
GS                             & 256               & 65 & $10^{-6}$\\
PRC                             & 512              & 60 &$10^{-5}$ \\
TR                       & 0 & - & 0.01 \\
\midrule
VideoShield &512    &60      &$10^{-5}$ \\
VideoMark &32       &75& 0.01 \\
\bottomrule
\end{tabular}
\end{table}

\subsection{CQS Image Quality Metric}
To provide a holistic assessment of the recovered image quality, we propose the CQS in \Cref{sec:exp_set}, which aggregates heterogeneous metrics (i.e., PSNR, SSIM, LPIPS, FID, and BRISQUE) into a unified score. The calculation is: 

\begin{equation}
\begin{aligned}
    \mathrm{CQS}&= \sum_{m} S_m, \\
S_m& =
\begin{cases}
\dfrac{v_m - v_m^{\min}}{v_m^{\max} - v_m^{\min}},
& m \in \{\mathrm{PSNR},\,\mathrm{SSIM}\}, \\
\dfrac{v_m^{\max} - v_m}{v_m^{\max} - v_m^{\min}},
& m \in \{\mathrm{LPIPS},\,\mathrm{FID},\,\mathrm{BRISQUE}\},
\end{cases}
\end{aligned}
\end{equation}
where $v_m$ denotes the value of metric $m$ and $v_m^{\min},v_m^{\max}$ are its min/max over the compared methods, yielding a per-metric score $S_m\in[0,1]$ and an aggregate $\mathrm{CQS}\in[0,5]$.

\section{Qualitative Comparison} \label{sec:quali_analy}

\Cref{fig:img_exp} qualitatively demonstrates the perceptual superiority of \method and \methodA over existing baselines. Our attacks produce images that are nearly indistinguishable from the originals. In the video domain, they further maintain temporal consistency without noticeable inter-frame flickering (\Cref{fig:video}). This high fidelity arises from on-manifold latent manipulation, which introduces structured, semantically coherent perturbations.

\begin{figure*}[t]
    \centering
    \includegraphics[width=0.9\linewidth]{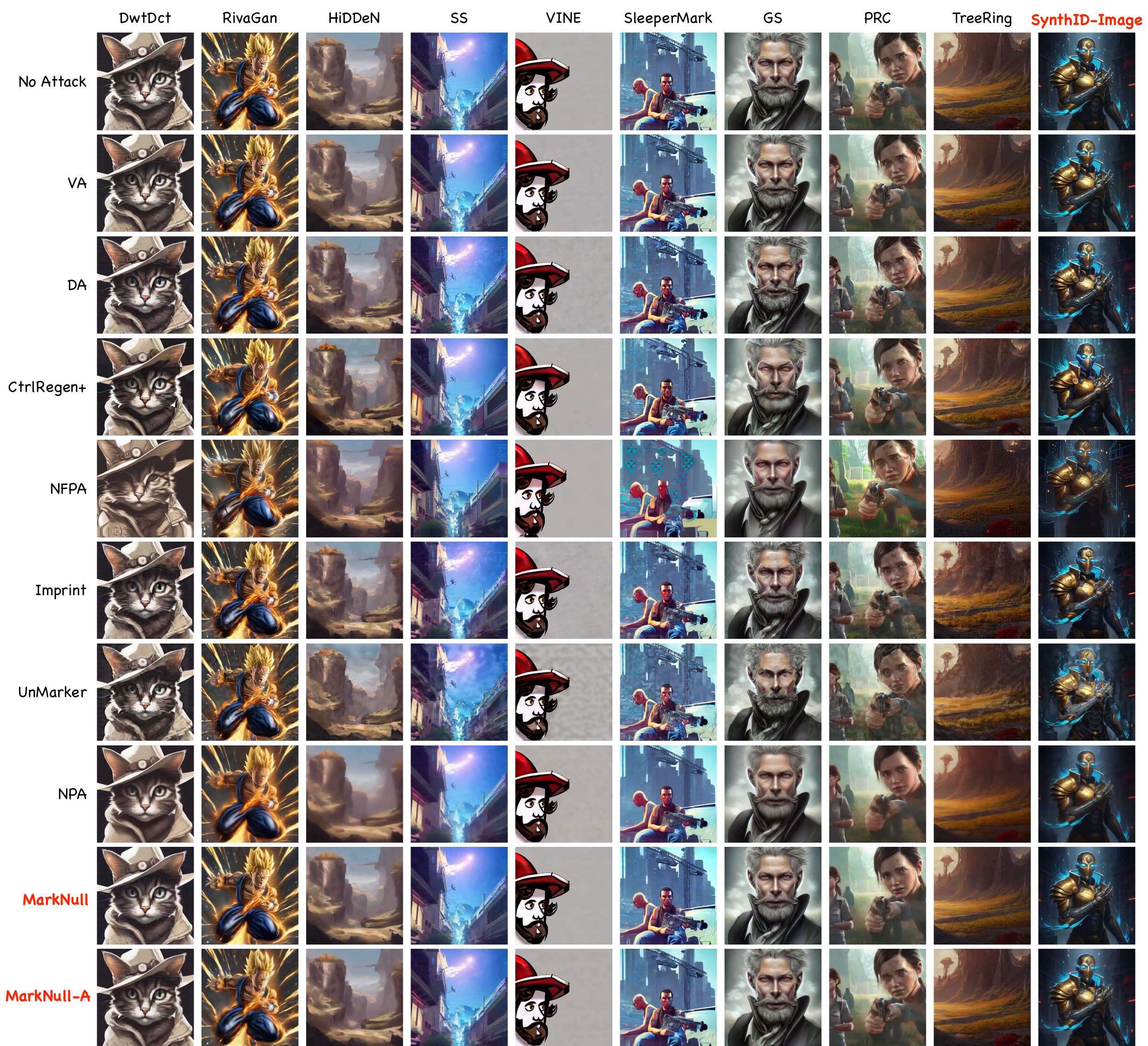}
    \caption{Qualitative examples of different attacks against baselines for watermark removal.}
    \label{fig:img_exp}
\end{figure*}

\begin{figure*}[t]
    \centering
    \includegraphics[width=0.9\linewidth]{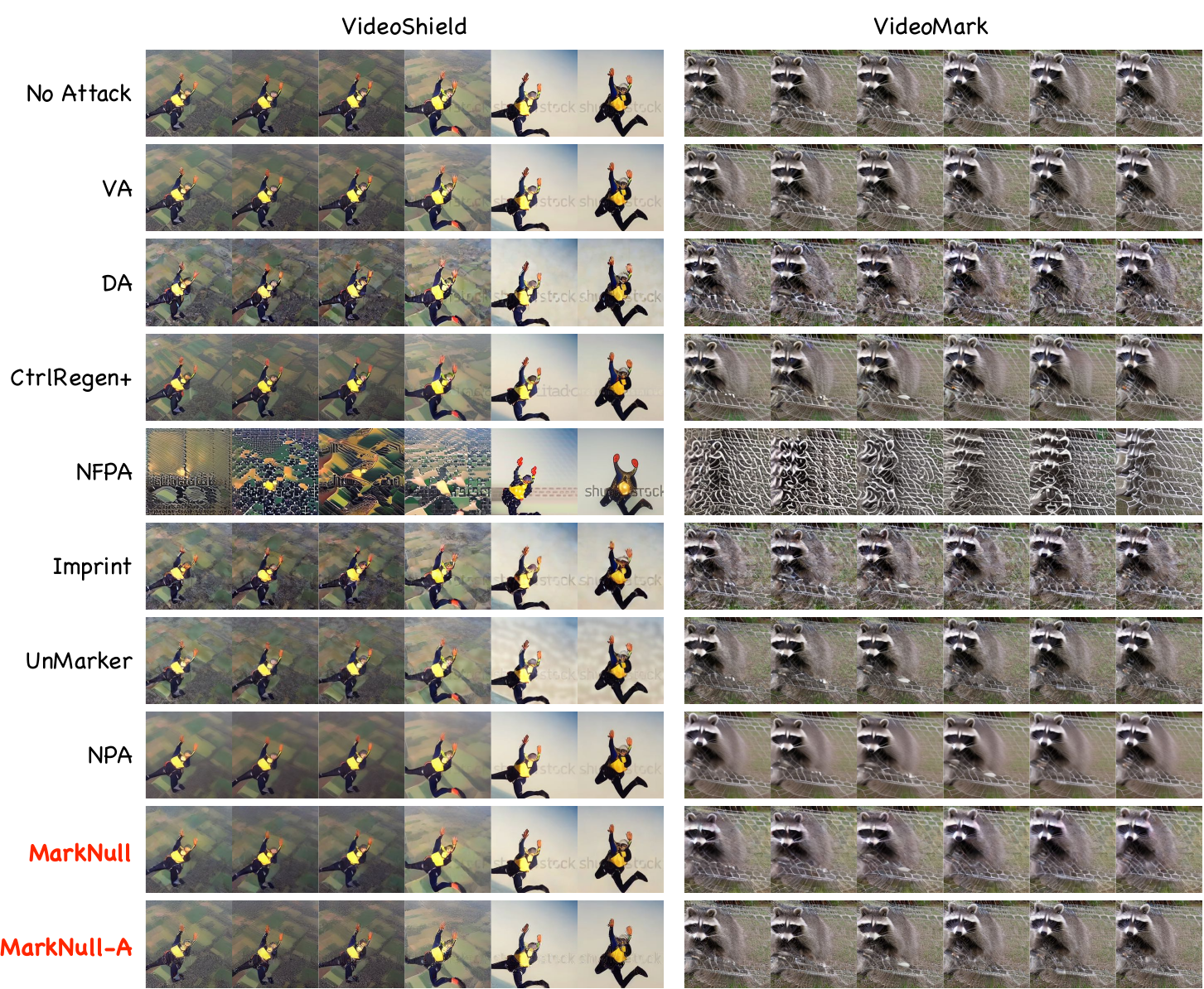}
    \caption{Qualitative examples of different attacks against VideoShield and VideoMark.}
    \label{fig:video}
\end{figure*}

\section{Discussion}
\subsection{Generalization Across Image Resolutions}

We discuss whether \method and \methodA can generalize across image resolutions. Although the main experiments are conducted at $512\times512$, the proposed attacks are not inherently resolution-specific. This is because \method operates in the latent space of latent diffusion models rather than directly in the pixel space. When the image resolution changes, the VAE encoder maps the image into a latent tensor with proportionally scaled spatial dimensions. The NLAS measures the normalized angular dependency between $z_0$ and $z_T$, and thus does not rely on a fixed image size. Therefore, the same formulation can be directly applied to images of different resolutions. \methodA's Restormer is also resolution-agnostic.

As shown in \cref{tab:size_generalization}, both \method and \methodA remain effective across different image sizes, reducing BA toward the random-guessing level while maintaining low LPIPS for representative schemes, including PRC~\cite{gunnPRC} and SS~\cite{fernandez2023stable}. This suggests that the attacks are not overfitted to the default $512\times512$ setting. We note that higher resolutions may increase the computational cost of iterative optimization for \method and may introduce mild distribution shifts for \methodA when the training and testing resolutions differ substantially. However, these factors mainly affect efficiency or small performance variations, rather than the validity of the attack itself. Overall, the designs of \method and \methodA provide favorable generalization across image resolutions.

\begin{table}[h]
\centering
\caption{Size generalization of \method and \methodA under three image resolutions}
\renewcommand{\arraystretch}{0.9}
\label{tab:size_generalization}
\resizebox{\linewidth}{!}{

\begin{tabular}{c c cc cc}

\toprule
\multirow{2}{*}{\textbf{Resolution}} 
& \multirow{2}{*}{\textbf{Method}} 
& \multicolumn{2}{c}{\textbf{\method}} 
& \multicolumn{2}{c}{\textbf{\methodA}} \\
\cmidrule(lr){3-4} \cmidrule(lr){5-6}
& & BA [\%] & LPIPS  
  & BA [\%] & LPIPS \\
\midrule
\multirow{2}{*}{$128 \times 128$} 
& PRC & 51.56 & 0.28 & 49.24 &  0.31 \\
& SS  & 38.09 & 0.23 & 46.20 &    0.29 \\
\midrule
\multirow{2}{*}{$256 \times 256$}
& PRC & 50.04 & 0.24 & 49.53 & 0.25 \\
 & SS  & 43.36 & 0.19 & 46.16 & 0.21 \\
\midrule
\multirow{2}{*}{$768 \times 768$} 
& PRC & 49.40 &0.20  & 49.00 & 0.23 \\
& SS  & 50.52 &0.14  & 46.68 & 0.23 \\
\bottomrule
\end{tabular}
}
\end{table}

\subsection{Transferability Across Diverse Backbones}

We further examine whether \method can transfer across different generative backbones. In the main experiments, we consider SD models as both the proxy and target backbones. A natural question is whether the attack remains effective when the proxy model is no longer from the SD family, while the target model is still SD-based.

To answer this question, we evaluate \textit{FLUX.1-dev}\footnote{\url{https://huggingface.co/black-forest-labs/FLUX.1-dev}} as the proxy model against SD2.1 as the target model. This setting introduces a substantial architectural divergence between the proxy and target models, making it more challenging than transfer within the same SD family. Nevertheless, \method remains effective under this setting, reducing BA to 56.95\% for GS, 52.65\% for PRC, and 46.67\% for SS. These results show that \method does not rely on a proxy model with the same architecture as the target model, and can transfer even when the proxy and target are built upon different generative backbones. We acknowledge that cross-backbone transfer incurs a non-negligible fidelity cost. Using \textit{FLUX.1-dev} as the proxy widens the proxy-target distribution gap relative to SD-family proxies, revealing a transferability–fidelity trade-off: \method remains effective across heterogeneous backbones, but stronger architectural mismatch can introduce more visible image changes.

This transferability stems from the fact that \method does not exploit model-specific implementation details. Instead, it suppresses watermark-related statistical dependencies through latent-space manipulation. Therefore, as long as the proxy model provides a generative prior that supports realistic image reconstruction and latent perturbation, the optimized perturbation can still generalize to the target watermarking pipeline. The FLUX.1-dev-to-SD2.1 results further confirm that \method captures a transferable vulnerability of current watermarking schemes, rather than overfitting to a particular SD backbone.

\section{Watermarking Baselines Configuration}\label{sec:wm_baseline}

\begin{itemize}
    \item DwtDct \cite{navas2008dwt}: DwtDct is a traditional \textit{post-hoc} watermarking method that embeds signals in the frequency domain combining DWT and DCT. We implement it via the \texttt{invisible-watermark} Python library \footnote{\url{https://github.com/ShieldMnt/invisible-watermark}}.
    
    \item RivaGan \cite{zhang2019robust}: RivaGAN is an encoder–decoder-based watermarking method proposed for video watermarking, which can be naturally extended to robust image watermarking. It employs an attention-based encoder-decoder to embed arbitrary payloads while maintaining visual fidelity. Its dual-discriminator design jointly enforces perceptual quality and watermark robustness through adversarial training. We implement it via the \texttt{invisible-watermark} Python library \footnote{\url{https://github.com/ShieldMnt/invisible-watermark}}.
    
    \item HiDDeN \cite{zhu2018hidden}: HiDDeN is an end-to-end trainable scheme for image watermarking, combining an encoder, decoder, and adversarial discriminator. It jointly optimizes imperceptibility, robustness, and payload capacity by minimizing reconstruction loss and adversarial detectability. To enhance robustness, it simulates differentiable distortions (e.g., blurring, cropping, JPEG) during training. In our implementation, we use the code and the pretrained model published in \footnote{\url{https://github.com/facebookresearch/stable_signature/tree/main/hidden}}.
    
    \item StableSignature (SS) \cite{fernandez2023stable}: SS is an in-generation watermarking method that embeds the watermark directly during the image synthesis process. It adopts the HiDDeN~\cite{zhu2018hidden} decoder as the watermark extractor and fine-tunes the VAE decoder of the LDM to ensure that the watermark remains decodable from the generated outputs. In our implementation, we follow the official code and utilize the pretrained model provided by the original authors \footnote{\url{https://github.com/facebookresearch/stable_signature}}.
    
    \item VINE \cite{lu2024robust}: VINE leverages a pretrained one-step diffusion model, SDXL-Turbo, as the watermark encoder. It employs a condition adaptor to fuse the watermark with the input image and introduces zero-convolution layers with skip connections to the VAE to ensure that the embedded watermark is both imperceptible and robust against advanced image editing. In our implementation, we follow the official code and utilize the pretrained model provided by the original authors \footnote{\url{https://github.com/Shilin-LU/VINE}}.
    
    \item SleeperMark \cite{wang2024sleepermark}: SleeperMark is designed to protect the copyright of T2I LDMs against unauthorized downstream fine-tuning. It injects a message-embedding backdoor into the diffusion backbone that is activated by a specific text trigger. For LDMs, the watermark is extracted from the latent space to achieve inherent robustness against image distortions with minimal perceptual impact. In our implementation, we follow the official code and utilize the pretrained model provided by the original authors \footnote{\url{https://github.com/taco-group/SleeperMark}}.
    
    \item Gaussian Shading (GS) \cite{yang2024gaussian}: GS is an initial-noised watermarking method that embeds a $k$-bit message by controlling the initial latent noise via distribution-preserving sampling. It first diffuses the binary watermark and then encrypts it with a stream key into a uniformly random bitstream, enabling the randomized watermark to drive a sampling procedure that produces a watermarked initial latent noise $z^w_T$ while preserving the standard Gaussian prior. The model then follows the standard denoising pipeline to obtain the final watermarked image.  For extraction, the suspect image is encoded back to the latent space, and DDIM inversion is used to estimate the corresponding $z^w_T$; the inverse sampling map recovers the encrypted bitstream, which is then decrypted and inverse-diffused, and a majority vote over replicated copies reconstructs the final watermark sequence. In our implementation, we follow the official code provided by the original authors \footnote{\url{https://github.com/bsmhmmlf/Gaussian-Shading}}.
    
    \item PRC \cite{gunnPRC}: PRC is an initial-noise-based watermarking scheme for diffusion models that embeds a secret key directly into the sign pattern of the initial noise to ensure undetectability. It generates a pseudorandom codeword and flips the signs of the sampled Gaussian vector to match this codeword while preserving the magnitude, thereby maintaining the marginal distribution of the latent. In our implementation, we follow the official code provided by the original authors \footnote{\url{https://github.com/XuandongZhao/PRC-Watermark}}.

    \item TreeRing \cite{wen2023tree}: TreeRing is the earliest initial-noise-based diffusion watermarking method, but it only supports zero-bit verification. Instead of encoding a binary watermark, TreeRing embeds a fixed ring watermark by modifying the initial noise with a predefined pattern. During verification, the suspect image is inverted via DDIM inversion to recover the corresponding initial noise, and the presence of the ring pattern is then tested in the recovered noise space. In our implementation, we follow the official code provided by the original authors \footnote{\url{https://github.com/YuxinWenRick/tree-ring-watermark}}.
  
\end{itemize}

\section{Attack Baselines Configuration} \label{sec:att_baseline}

\begin{itemize}
    \item Noise: Gaussian noise injection is a common additive distortion that perturbs watermark patterns by introducing pixel-wise random fluctuations. In our implementation, each image is first normalized to $[0,1]$, then i.i.d. Gaussian noise is added, where \texttt{std=0.05} controls the noise intensity. 
    
    \item Blur: Gaussian blur is a smoothing attack that degrades watermarks by attenuating high-frequency image components. This low-pass filtering suppresses fine textures and sharp edges, which can weaken watermark cues especially those encoded in subtle spatial details. In our implementation, we apply \texttt{cv2.GaussianBlur} to each input image with a kernel size of $5\times5$ and Gaussian standard deviation $\sigma=1$.
    
    \item Contrast: Contrast adjustment is a common photometric distortion that modifies the dynamic range of pixel intensities, thereby perturbing watermark cues while leaving the image geometry unchanged. In \Cref{tab:attack_vs_methods}, we apply \texttt{PIL.ImageEnhance.Contrast} with a contrast factor of 0.5, which reduces contrast by compressing intensity variations around the mean.
    
    \item JPEG: JPEG compression is a widely used lossy post-processing operation, controlled by a quality parameter that determines the compression strength. Lower quality values induce stronger quantization and larger information loss, which can further impair watermark recoverability. In \Cref{tab:attack_vs_methods}, we apply JPEG compression with the quality factor set to 80.
    
    \item Brightness: Brightness adjustment is a common photometric distortion that perturbs watermark signals by uniformly scaling image luminance while largely preserving geometric structure. In our implementation, we use \texttt{PIL.ImageEnhance.Brightness} to re-render each image with a brightness factor of 0.5, i.e., substantially darkening the image.
    
    \item VA \cite{an2024waves}: VA is a regeneration attack that removes watermarks by re-synthesizing the image through a VAE. Concretely, the watermarked image is first encoded into a latent representation by a neural VAE encoder and then decoded back to the pixel space by the corresponding neural decoder. In our experiment, we adopt \texttt{compressai.zoo.cheng2020\_anchor} to implement VA following \cite{zhao2024invisible}. The compression factor governs the attack strength, where smaller values impose a tighter bottleneck and thus induce more aggressive distortion. In this paper, we set the compression factor to $3$.

    \item DA \cite{zhao2024invisible}: DA is a regeneration attack that removes watermarks by perturbing the watermarked image with a forward noising process and then reconstructing it via the reverse denoising trajectory. Using more noising steps injects stronger corruption, which typically yields a higher watermark removal rate at the cost of greater perceptual deviation. In our implementation, we follow the official code provided by the original authors \footnote{\url{https://github.com/XuandongZhao/WatermarkAttacker}}. In our experiment, we instantiate DA with SD1.5 and set the number of noise steps to 60.

    \item CtrlRegen$+$ \cite{liu2024image}: CtrlRegen$+$ is an adjustable watermark removal method that leverages a controllable regeneration process to eliminate embedded signals. It proceeds by projecting the watermarked image into the latent space and introducing noise over a specified number of steps to produce a noisy latent representation. Subsequently, it employs a controllable diffusion model to denoise and reconstruct the image, utilizing semantic and spatial features extracted from the original input to ensure high visual fidelity and consistency. In our implementation, we follow the official code provided by the original authors \footnote{\url{https://github.com/yepengliu/CtrlRegen}}.
    
    \item NFPA \cite{qiufuture}: NFPA is a regeneration attack that removes watermarks by generating a new, semantically consistent image from the input via a \textit{next-frame prediction} style diffusion process. It first encodes the input image into the latent representation and applies DDIM inversion to obtain an inverted Latent as the generation starting point. In our implementation, we follow the official code provided by the original authors \footnote{\url{https://github.com/1249748036/NFPA/}}. 
    
    \item Imprint \cite{muller2025black}: Imprint is a removal attack for initial-noise-based watermarking, operating with a proxy model. It first encodes the watermarked image into the proxy model’s latent space and applies inversion to recover an estimate of the corresponding initial noise. It then optimizes a small perturbation to the latent representation via gradient-based updates, with an objective that drives the re-inverted noise away from the original estimate. The attack budget is governed by the number of optimization iterations: using more steps typically strengthens watermark suppression, but may also amplify perceptual distortion and increase computational overhead. In our implementation, we follow the official code provided by the original authors \footnote{\url{https://github.com/and-mill/semantic-forgery}}, use 50 optimization steps, and the proxy model is SD1.5.

    \item UnMarker \cite{kassis2025unmarker}: UnMarker is a universal black-box attack method without requiring detector feedback or specific knowledge of the watermarking algorithm. It operates on the insight that robust watermarks inevitably modify spectral amplitudes, and thus employs two novel adversarial optimization strategies to specifically disrupt these spectral features, effectively erasing the watermark signals while preserving image quality.  In our implementation, we follow the official code provided by the original authors \footnote{\url{https://github.com/andrekassis/ai-watermark}} and apply both adversarial strategies against all watermarking schemes.
    
    \item NPA \cite{goren2025noiseprints}: NoisePrints is a lightweight, distortion-free watermarking scheme originally. It also makes the simple observation that the initial noise can be highly correlated with the generated content. In the appendix of \cite{goren2025noiseprints}, the authors further introduce an attack, which we refer to as NPA. This attack uses a proxy model to compute the cosine similarity between the image and the recovered noise, and performs an iterative optimization that updates a small perturbation in the latent space. The objective is to drive the adversarial latent representation away from the initial noise by minimizing its cosine similarity. Since the authors do not release official code, we implement the attack based on the paper description; in our setting, the number of optimization steps is set to $100$.  
\end{itemize}

\end{document}